\documentclass[12pt,amsfonts]{article}
\usepackage{mathrsfs}
\usepackage{amsfonts,amsmath,amssymb}
\usepackage{multirow}
\usepackage{longtable}
\usepackage{graphicx}
\usepackage{epstopdf}
\usepackage{enumerate}
\usepackage{stmaryrd}
\usepackage{hyperref}

\newtheorem{theorem}{Theorem}[section]
\newtheorem{lemma}[theorem]{Lemma}
\newtheorem{proposition}[theorem]{Proposition}
\newtheorem{corollary}[theorem]{Corollary}

\newtheorem{construction}[theorem]{Construction}

\def\pbd#1,#2/{$\left(#1,\left\{5,#2^\star\right\}\!,1\right)$-PBD} 

\newenvironment{proof}{\noindent{\bf Proof:}}{\quad\hfill$\Box$\vspace*{2ex}}

\newcommand{\A}{{\mathcal A}}
\newcommand{\B}{{\mathcal B}}
\newcommand{\C}{{\mathcal C}}
\newcommand{\D}{{\mathcal D}}

\newcommand{\G}{{\mathcal G}}

\newcommand{\vu}{{u}}
\newcommand{\vv}{{v}}

\newcommand{\supp}{{\rm supp}}

\newcommand{\bbZ}{{\mathbb{Z}}}

\begin{document}

\title{Optimal Ternary Constant-Composition Codes with Weight Four and Distance Six}

\author{Hengjia Wei$^{\text{a,b}}$, Hui Zhang$^{\text{b}}$, Mingzhi Zhu$^{\text{b}}$ and Gennian Ge$^{\text{a,c,}}$\thanks{Corresponding author. Email address: gnge@zju.edu.cn.  Research supported by the National Natural Science Foundation of China under Grant No.~61171198 and  Grant No.~11431003, and Zhejiang Provincial Natural Science Foundation of China under Grant No.~LZ13A010001.}\\
  \footnotesize $^{\text{a}}$ School of Mathematical Sciences, Capital Normal University, Beijing, 100048, China\\
  \footnotesize $^{\text{b}}$ Department of Mathematics, Zhejiang University, Hangzhou 310027, Zhejiang, China\\
\footnotesize $^{\text{c}}$ Beijing Center for Mathematics and Information
Interdisciplinary Sciences,
Beijing, 100048, China.
}

\date{}
\maketitle

\vskip .1 in

\begin{abstract}
The sizes of optimal constant-composition codes of
weight three have been determined by Chee, Ge and Ling with  four
cases in doubt. Group divisible codes played an important role in
their constructions. In this paper, we study the problem of
constructing optimal ternary constant-composition codes with
Hamming weight four and minimum distance six.  The problem is solved with a small number of lengths undetermined.
The previously known results are those with code length no greater than 10.


\medskip
\noindent {{\it Key words and phrases\/}: Constant-composition codes, group divisible codes,
ternary codes}
\smallskip

\noindent {{\it AMS subject classifications\/}: Primary 05B05, 94B60.}
\smallskip
\end{abstract}

\section{Introduction}
Constant-composition codes (CCCs) are a special type of
constant-weight codes  (CWCs) which  are important in coding theory. The
class of constant-composition codes includes the important
permutation codes and has attracted recent interest due to their
numerous applications, such as in determining the zero error
decision feedback capacity of discrete memoryless channels
\cite{Telatar:1990}, multiple-access communications
\cite{D'yachkov:1984}, spherical codes for modulation
\cite{Ericson:1995}, DNA codes
\cite{King:2003,Milenkovic:2006,Chee:2008}, powerline communications
\cite{Chu:2004,Colbourn:2004}, frequency hopping \cite{Chu:2006}, frequency permutation arrays \cite{Huczynska:2006}, and
coding for bandwidth-limited channels \cite{Costello:2007}.

Systematic study began in late 1990's \cite{Bogdanova:1998,Svanstrom:1999,Bogdanova_Ocetarova:1998}. Today, various methods have been applied to
the problem of determining the maximum size of a
constant-composition code, such as computer search methods \cite{Bogdanova:2003}, packing designs
\cite{Ding:2005(2),Chu:2006,Ding:2006,Yan:2008,YinTang:2008,Wen:2009,Yan:2009,Huczynska:2010},
tournament designs \cite{YinYan:2008}, polynomials and nonlinear
functions
\cite{Ding:2005,Ding:2005(3),Chu:2006,Ding:2008,YDing:2008},
difference triangle sets \cite{Chee:2010}, PBD-closure methods
\cite{Yeow:2007,Yeow:2008} and some other methods \cite{Svanstrom:2000,Luo:2003}.

In the paper of Svanstr{\"o}m et al. \cite{Svanstrom:2002}, some
methods for providing upper and lower bounds on the maximum size
$A_3(n,d,\overline{w})$ of a ternary code with length $n$, minimum
Hamming distance $d$, and constant composition $\overline{w}$ were
presented, and a table of exact values or bounds for $A_3(n,d,\overline{w})$ in the
range $n\leq 10$ was also given there. Here we list the exact values
of $A_3(n,6,[2,2])$ and $A_3(n,6,[3,1])$ for codes with length no
greater than 10 in Table~\ref{presult}.

The sizes of optimal ternary CCCs
with weight three have been determined completely by Chee, Ge and
Ling in \cite{Yeow:2008}. The sizes of optimal ternary
CCCs with weight four and distance five
have been determined completely by Gao and Ge in \cite{Gao:2010}.

In this paper, we will concentrate our attention on ternary CCCs with  weight four and  distance six. We shall use group divisible codes as the main tools, which were first introduced
by Chee et al. in \cite{Yeow:2008} and have been
shown to be useful in recursive constructions of CWCs and CCCs. The
article is organized as follows. Section II provides some basic
definitions and results on combinatorial designs and coding
theory. In Sections III and IV, we focus on the determination for the exact values of $A_3(n,6,[2,2])$ and $A_3(n,6,[3,1])$ respectively. A brief conclusion is presented in Section V.


\begin{table}
\caption{Values of $A_3(n,6,[2,2])$ and $A_3(n,6,[3,1])$ for $n\leq 10$} \label{presult}
\centering
\begin{tabular}{c|c c c c c c c}
\hline
$n$              & $4$ & $5$ & $6$ & $7$ & $8$ & $9$ & $10$ \\
\hline
$A_3(n,6,[2,2])$ & $1$ & $1$ & $3$ & $3$ & $5$ & $9$ & $15$ \\
\hline
$A_3(n,6,[3,1])$ & $1$ & $1$ & $2$ & $2$ & $4$ & $6$ & $10$ \\
\hline
\end{tabular}
\end{table}

%
%
%

\section{Preliminaries}
\vskip 10pt

\subsection{Definitions and Notations}
The set of integers $\{i,i+1,\ldots,j\}$ is denoted by $[i,j]$. When $i=0$ and $j=q-1$, the set is also denoted by $I_q$. The ring $\bbZ/q\bbZ$ is denoted by $\bbZ_q$. The notation $\Lbag \cdot
\Rbag$ is used for multisets.

All sets considered in this paper are finite if not obviously
infinite. If $X$ and $R$ are finite sets, $R^X$ denotes the set of
vectors of length $|X|$, where each component of a vector $\vu\in
R^X$ has value in $R$ and is indexed by an element of $X$, that is,
$\vu=(\vu_x)_{x\in X}$, and $\vu_x\in R$ for each $x\in X$. A
{\it $q$-ary code of length $n$} is a set $\C \subseteq \bbZ_q^X$ for some
$X$ with size $n$. The elements of $\C$ are called {\it codewords}. The {\it support} of a vector $\vu \in \bbZ_q^X$, denoted $\supp(\vu)$, is the set $\{x \in X : \vu_x \not= 0\}$. The
{\it Hamming norm} or the {\it Hamming weight} of a vector $\vu\in\bbZ_q^X$ is
defined as $\|\vu\|=| \supp(\vu)|$. The distance
induced by this norm is called the {\it Hamming distance}, denoted $d_H$,
so that $d_H(\vu,\vv)=\| \vu-\vv \|$, for $\vu,\vv\in\bbZ_q^X$.
A code $\C$ is said to have {\it minimum distance $d$} if $d_H(\vu,\vv)\geq d$ for
all distinct $\vu,\vv\in\C$. The {\it composition} of a vector $\vu\in \bbZ_q^X$ is the tuple
$\overline{w}=[w_1,\ldots,w_{q-1}]$, where $w_j=|\{x\in X:
\vu_x=j\}|$. A code $\C$ is said to have {\it constant weight $w$} if every codeword in $\C$ has weight $w$, and is said to have {\it constant composition $\overline{w}$} if every codeword in $\C$ has composition $\overline{w}$. Hence, every constant-composition code is
a constant-weight code. We refer to a $q$-ary code of length $n$,
distance $d$, and constant weight $w$ as an $(n,d,w)_q$-code. If in addition, the code has constant composition $\overline{w}$, then it is referred to as an $(n,d,\overline{w})_q$-code.
The maximum size of an $(n,d,w)_q$-code is denoted as $A_q(n,d,w)$ and that of an $(n,d,\overline{w})_q$-code is denoted as $A_q(n,d,\overline{w})$.
Any $(n,d,w)_q$-code or $(n,d,\overline{w})_q$-code achieving the maximum size is called
{\it optimal}.

The following operations do not affect distance
and weight properties of an $(n,d,\overline{w})_q$-code:
\begin{enumerate}[(i)]
\item reordering the components of $\overline{w}$, and
\item deleting zero components of $\overline{w}$.
\end{enumerate}
Consequently, throughout this paper, we restrict our attention to
those compositions $\overline{w}=[w_1,\ldots,w_{q-1}]$, where
$w_1\geq \cdots\geq w_{q-1}\geq 1$.

Suppose $\vu\in\bbZ_q^X$ is a codeword of an
$(n,d,\overline{w})_q$-code, where
$\overline{w}=[w_1,\ldots,w_{q-1}]$. Let $w=\sum_{i=1}^{q-1} w_i$.
We can represent $\vu$ equivalently as a $w$-tuple $\langle
a_1, a_2, \ldots, a_w\rangle\in X^w$, where
\begin{align*}
\vu_{a_1}=\cdots=\vu_{a_{w_1}} &=1, \\
\vu_{a_{w_1+1}}=\cdots=\vu_{a_{w_1+w_2}} &=2, \\
\vdots \\
\vu_{a_{\sum_{i=1}^{q-2} w_i+1}} =\cdots=\vu_w &= q-1.
\end{align*}
Throughout this paper, we shall often represent codewords of
constant-composition codes in this form. This has the advantage of
being more succinct and more flexible in manipulation.

\subsection{Upper Bounds}
\vskip 10pt
For constant-composition codes, we have

\vskip 10pt

\begin{lemma}[Chee et al. \cite{Yeow:2008,Chee:2010}]\label{bound1}
\begin{equation*}
\begin{split}
& A_q(n,d,[w_1,\ldots,w_{q-1}])=\\
& \begin{cases} \binom{n}{\sum_{i=1}^{q-1}w_i}
\binom{\sum_{i=1}^{q-1}w_i}{w_1,\ldots,w_{q-1}}, & \text{if $d \leq 2$}\\
\lfloor n/w_1\rfloor, &  \text{if $d=2\sum_{i=1}^{q-1}w_i-1$ and $n$ is sufficiently large} \\
\left\lfloor \frac{n}{\sum_{i=1}^{q-1}w_i} \right\rfloor, & \text{if
$d=2
\sum_{i=1}^{q-1}w_i$}\\
1, & \text{if $d \geq 2\sum_{i=1}^{q-1}w_i+1$}.
\end{cases}
\end{split}
\end{equation*}
\end{lemma}
\vskip 10pt


The following Johnson-type bound has been proven for constant-composition codes.

\begin{theorem}[Svanstr{\"o}m et al. \cite{Svanstrom:2002}]\label{bound2}
\begin{equation*}
A_q(n,d,[w_1,\ldots,w_{q-1}]) \leq \frac{n}{w_1}A_q(n-1,d,[w_1-1,
\ldots ,w_{q-1}]).
\end{equation*}
\end{theorem}

As a consequence of Lemma~\ref{bound1} and Theorem~\ref{bound2}, we have
the following result.

\begin{corollary}
\begin{equation*}
\begin{split}
A_3(n,6,[2,2]) \leq \left\lfloor \frac{n}{2} \left\lfloor
\frac{n-1}{3} \right\rfloor \right\rfloor,\\
\end{split}
\end{equation*}
\begin{equation*}
A_3(n,6,[3,1]) \leq \left\lfloor \frac{n}{3} \left\lfloor
\frac{n-1}{3} \right\rfloor \right\rfloor.
\end{equation*}
\end{corollary}

For most cases, we will show that the above Johnson-type bound is tight. However, for the cases $n\equiv 4,5,7\pmod{9}$ and $\overline{w}=[3,1]$, other arguments can give better bounds.

\begin{lemma}
\begin{equation*}
A_3(9t+4,6,[3,1]) \leq 9t^2+6t+1+\lfloor\frac{t}{4}\rfloor,\\
\end{equation*}
\begin{equation*}
A_3(9t+5,6,[3,1]) \leq 9t^2+7t+1+\lfloor\frac{t+1}{4}\rfloor,\\
\end{equation*}
\begin{equation*}
A_3(9t+7,6,[3,1]) \leq 9t^2+11t+3+\lfloor\frac{t+1}{2}\rfloor.
\end{equation*}
\end{lemma}
\begin{proof}
Let $\cal C$ be a code of length $n$ with composition $[3,1]$, minimum distance six and $M$ codewords. Let ${\cal C}_i^{1}$ denote the set of codewords with $1$ in position $i$ and let ${\cal C}_i^{2}$ denote the set of codewords with $2$ in position $i$. Let $x_i,y_i$ be the sizes of ${\cal C}_i^{1}$ and ${\cal C}_i^{2}$ respectively. Count the number of nonzero symbols in the code in two ways to get $4M = \sum_i (x_i +y_i)$. We want to bound the value of $x_i + y_i$.
Consider a fixed position $i$ of $\cal C$. Since the minimum distance of $\cal C$ is six, the remaining nonzero symbols of the codewords in ${\cal C}_i^{1}$ should lie in different positions. So we have $3x_i \leq n-1$. Now, we consider a codeword in ${\cal C}_i^{2}$. Such a codeword cannot have its $1$s in any of the $2x_i$ positions where a codeword of ${\cal C}_i^{1}$ has a $1$, as this would
give a minimum distance smaller than six. Furthermore, the $1$s of two codewords in ${\cal C}_i^{2}$ cannot overlap, so we have $2x_i+3y_i\leq n-1$.

When $n=9t+4$, it follows from $2x_i+3y_i\leq 9t+3$ that $x_i+y_i\leq 3t+1+\lfloor \frac{x_i}{3}\rfloor$. Noting that $3x_i \leq 9t+3$, we can get $x_i+y_i\leq 4t+1$. Thus $M \leq \lfloor \frac {(9t+4)(4t+1)}{4}\rfloor$. We then obtain the first inequation. The other inequations can be obtained by similar arguments.
\end{proof}

In the rest of this paper, we use the notation $U(n,6,[2,2])=\left\lfloor \frac{n}{2} \left\lfloor \frac{n-1}{3}\right\rfloor \right\rfloor$ as the upper bound for the maximum size of an $(n,6,[2,2])_3$-code. For the composition $[3,1]$, when $n\equiv 0,1,2,3,6,8\pmod{9}$, we denote $U(n,6,[3,1])=\left\lfloor \frac{n}{3} \left\lfloor \frac{n-1}{3}\right\rfloor \right\rfloor$; for $n\equiv 4,5,7\pmod{9}$, write $n=9t+i$ with $i=4$, $5$ or $7$ and denote
\begin{equation*}
U(9t+4,6,[3,1]) = 9t^2+6t+1+\lfloor\frac{t}{4}\rfloor,\\
\end{equation*}
\begin{equation*}
U(9t+5,6,[3,1]) = 9t^2+7t+1+\lfloor\frac{t+1}{4}\rfloor,\\
\end{equation*}
\begin{equation*}
U(9t+7,6,[3,1]) = 9t^2+11t+3+\lfloor\frac{t+1}{2}\rfloor.
\end{equation*}

%


\subsection{Designs}

Our recursive construction is based on some combinatorial structures
in design theory. The most important tools are pairwise balanced
designs (PBDs) and group divisible designs (GDDs).

Let $K$ be a subset of positive integers and let $\lambda$ be a
positive integer. A {\it pairwise balanced design} ($(v, K, \lambda)$-PBD or $(K, \lambda)$-PBD of order $v$) is a pair
($X,\B$), where $X$ is a finite set ({\it the point set}) of cardinality
$v$ and $\B$ is a family of subsets ({\it blocks}) of $X$ that satisfy
(1) if $B\in \B$, then $|B|\in K$ and (2) every pair of distinct
elements of $X$ occurs in exactly $\lambda$ blocks of $\B$. The
integer $\lambda$ is the index of the PBD.

\begin{theorem}[Abel et al. \cite{AbelBG:2007}]
\label{PBD4-9} For any integer $v\geq 10$, a $(v,\{4,5,6,7,8,9\},1)$-PBD
exists with exceptions $v\in \{10,11,12,$ $14,15,18,19,23\}$.
\end{theorem}

\begin{theorem}[Abel et al. \cite{AbelBG:2007}]
\label{PBD5-9} For any integer $v\geq 10$, a
$(v,\{5,6,7,8,9\},1)$-PBD exists with exceptions $v\in
[10,20]\cup[22,24]\cup[27,29]\cup[32,34]$.
\end{theorem}

A {\it group divisible design} (GDD) is a triple $(X,{\cal G},{\cal B})$ where $X$ is a set of points, ${\cal G}$ is a partition of $X$ into {\it groups}, and ${\cal B}$ is a collection of subsets of $X$ called {\it blocks} such that any pair of distinct points from $X$ occurs either in some group or in exactly one block, but not both.  A $K$-GDD of type $g^{u_1}_1 g^{u_2}_2\ldots g^{u_s}_s$ is a GDD in which every block has size from the set $K$ and in which there are $u_i$ groups of size $g_i, i=1,2,\ldots,s$. When $K=\{k\}$, we simply write $k$ for $K$.  A {\it parallel
class} or {\it resolution class} is a collection of blocks that partitions
the point set of the design. A GDD is {\it resolvable} if the blocks of
the design can be partitioned into parallel classes. A resolvable
GDD is denoted by RGDD.

A $k$-GDD of type $m^k$ is also called a {\em transversal
design} and denoted by TD$(k,m)$.

\begin{theorem}[Abel et al. \cite{Abel:2007}]
\label{TD} Let $m$ be a positive integer. Then:
\begin{enumerate}
\item[i)] a TD$(4,m)$ exists if $m\not\in \{2,6\}$;
\item[ii)] a TD$(5,m)$ exists if $m\not\in \{2,3,6,10\}$;
\item[(iii)] a TD$(6,m)$ exists if $m\not\in\{2,3,4,6,10,14,18,22\}$;
\item[(iv)] a TD$(7,m)$ exists if $m\not\in\{2,3,4,5,6,10,14,15,18,20,22,26,30,34,38,46,60\}$;
\item[v)] a TD$(8,m)$ exists if $m\not\in \{2,3,4,5,6,10,12,14,15,18,20,21,22,26,28,30,33,34,35,38,$ $39,42,44,46,51,52,54,58,60,62,66,68,74\}$;
\item[vi)] a TD$(m+1,m)$ exists if $m$ is a prime power.
\end{enumerate}
\end{theorem}

A {\it double group divisible design} (DGDD) is a quadruple
$(X,{\cal H},{\cal G},{\cal B})$ where $X$ is a set of points,
${\cal H}$ and ${\cal G}$ are partitions of $X$ (into {\it holes} and
{\it groups}, respectively) and ${\cal B}$ is a collection of subsets of
$X$ ({\it blocks}) such that
\begin{enumerate}
\item [(i)]for  each block $B\in {\cal B}$ and each hole $H\in
{\cal H}, |B\cap H|\leq 1$, \item [(ii)] any pair of distinct
points from $X$ which are not in the same hole occurs either in
some group or in exactly one block, but not both.
\end{enumerate}
A $K$-DGDD of type $(g_1,h^v_1)^{u_1}(g_2,h_2^v)^{u_2} \ldots
(g_s,h_s^v)^{u_s}$ is a double group divisible design in which every
block has size from the set $K$ and in which there are $u_i$ groups
of size $g_i$, each of which intersects each of the $v$ holes in
$h_i$ points. A {\it modified group divisible design} $K$-MGDD of
type $g^u$ is a $K$-DGDD of type $(g,1^g)^u$.
%
%
%

\subsection{Group Divisible Codes}

Given $\vu \in \bbZ_q^X$ and $Y \subseteq X$, the
\emph{restriction of $\vu$ to $Y$}, written $\vu\mid_Y$, is the
vector $\vv \in \bbZ_q^X$ such that
$$\vv_x=
\begin{cases}
\vu_x, & \text{if $x \in Y$} \\
0, & \text{if $x \in X \backslash Y$}. \\
\end{cases}$$

A {\it group divisible code} (GDC) of distance $d$ is a triple
($X,\G,\C$), where $\G=\{G_1,\ldots,G_t\}$ is a partition of $X$
with cardinality $|X|=n$ and $\C\subseteq \bbZ_q^X$ is a $q$-ary
code of length $n$, such that $d_H(u,v) \ge d$ for each distinct
$u,v\in \C$, and $\|u|_{G_i}\|\le 1$ for each $u \in \C$, $1\le i \le
t$. Elements of $\G$ are called groups. We denote a GDC($X,\G,\C$)
of distance $d$ as $w$-GDC($d$) if $\C$ is of constant weight $w$.
If we want to emphasize the composition of the codewords, we denote
the GDC as $\overline{w}$-GDC($d$) when every $u\in \C$ has
composition $\overline{w}$. The type of a GDC($X,\G,\C$) is the
multiset $\Lbag |G|: G\in \G \Rbag$. As in the case of GDDs, the
exponential notation is used to describe the type of a GDC. The size
of a GDC($X,\G,\C$) is $|\C|$. Note that an
$(n,d,\overline{w})_q$-code with size $s$ is equivalent to a
$\overline{w}$-GDC($d$) of type $1^n$ with size $s$.

Constant-composition codes of larger orders can often be obtained from
GDCs via the following two constructions.

\begin{construction}[Filling in Groups, see \cite{Yeow:2008}]
\label{FillGroups} Let $d\leq 2(w-1)$. Suppose there exists a
$w$-${\rm GDC}(d)$ $(X,\G,\C)$ of type $g_1^{t_1}\cdots g_s^{t_s}$
with size $a$. Suppose further that for each $i$, $1\leq i\leq s$,
there exists a $(g_i,d,w)_q$-code $\C_i$ with size $b_i$, then there
exists a $(\sum_{i=1}^s t_ig_i,d,w)_q$-code $\C'$ with size
$a+\sum_{i=1}^s t_ib_i$. In particular, if $\C$ and $\C_i$, $1\leq
i\leq s$, are of constant composition $\overline{w}$, then $\C'$ is
also of constant composition $\overline{w}$.
\end{construction}

\begin{construction}[Adjoining $y$ Points, see \cite{Yeow:2008}]
\label{AdjoinPoints} Let  $d\leq 2(w-1)$ and  $y\in\bbZ_{\geq 0}$. Suppose there exists a
(master) $w$-${\rm GDC}(d)$ of type $g_1^{t_1}\cdots g_s^{t_s}$ with size $a$, and suppose the following (ingredients) also exist:
\begin{enumerate}[(i)]
\item a $(g_1+y,d,w)_q$-code with size $b$,
\item a $w$-${\rm GDC}(d)$ of type $1^{g_i} y^1$ with size $c_i$ for each $2\leq i\leq s$,
\item a $w$-${\rm GDC}(d)$ of type $1^{g_1} y^1$ with size $c_1$ if $t_1\geq 2$.
\end{enumerate}
Then, there exists a $(y+\sum_{i=1}^s t_ig_i,d,w)_q$-code with size $a+b+(t_1-1)c_1+\sum_{i=2}^{s} t_ic_i$. Furthermore, if the master and ingredient codes are of constant composition, then so is the resulting code.
\end{construction}

The  following two  constructions are useful for generating  GDCs of larger orders from smaller ones.

\begin{construction}[Fundamental Construction \cite{Yeow:2008}]
\label{FundCtr} Let $d\le 2(w-1)$, $\D=(X,\G,\A)$ be a (master) GDD,
and $\omega: X\rightarrow \bbZ_{\ge0}$ be a weight function. Suppose
that for each $A\in \A$, there exists an (ingredient) $w$-GDC($d$)
of type $\Lbag \omega(a): a\in A \Rbag$. Then there exists a
$w$-GDC($d$) $\D^{*}$ of type $\Lbag \sum_{x\in G}{\omega(x)}:G\in
\G \Rbag$. Furthermore, if the ingredient GDCs are of constant
composition $\overline{w}$, then $\D^{*}$ is also of constant
composition $\overline{w}$.
\end{construction}

\begin{construction}[Inflation Construction]
\label{Inflation} Let $d\leq 2(w-1)$. Suppose there exists a $w$-GDC($d$) of type
$g_1^{t_1}\dots g_s^{t_s}$ with size $a$. Suppose further that there
exists a TD($w,m$), then there exists a $w$-GDC($d$) of type
$(mg_1)^{t_1}\dots (mg_s)^{t_s}$ with size $am^2$. If the original
GDC is of constant composition $\overline{w}$, then so is the
derived GDC.
\end{construction}

\section{Determining the Value of $A_3(n,6,[2,2])$}

In this section,
 we focus on the determination for the exact values of  $A_3(n,6,[2,2])$ for all positive integers $n$.
We first construct some $[2,2]$-GDC$(6)$s to obtain the optimal $(n,6,[2,2])_3$-codes.

\subsection{Skew Room Frame Construction}

If $\{S_1,\dots,S_n\}$ is a partition of a set $S$, an
$\{S_1,\dots,S_n\}$-{\em Room frame} is an $|S|\times |S|$
array, $F$, indexed by $S$, satisfying:
\begin{itemize}
\item[1.] every cell of $F$ either is empty or contains an
    unordered pair of symbols of $S$,
\item[2.] the subarrays $S_i\times S_i$ are empty, for
    $1\leq i\leq n$ (these subarrays are {\em holes}),
\item[3.] each symbol $x\not\in S_i$ occurs once in row (or
    column) $s$ for any $s\in S_i$, and
\item[4.] the pairs occurring in $F$ are those $\{s, t\}$,
    where $(s, t)\in (S\times S)\backslash \bigcup_{i=1}^n
    (S_i\times S_i)$.
\end{itemize}

The {\it type} of an $\{S_1,\dots,S_n\}$-Room frame $F$ will be
the multiset $\Lbag |S_1|,\dots,| S_n|\Rbag$. We will say that $F$ has
type $t_1^{u_1} \dots t_k^{u_k}$ provided there are $u_j$
$S_i$'s of cardinality $t_j$, for $1\leq j \leq k$. A Room
frame is {\em skew} if cell $(i,j)$ is filled implies that cell
$(j,i)$ is empty. A Room frame of type $1^n$ is called a {\em
Room square}.


\begin{lemma} [\cite{CZ:1996}, \cite{ZG:2007}]
\label{SRF} The necessary conditions for the existence of a
skew Room frame of type $t^u$, namely, $u\geq 4$ and $t(u-1)$
is even, are also sufficient except for $(t,u)\in
\{(1,5),(2,4)\}$ and with possible exceptions:
\begin{itemize}
\item[(i)] $u=4$ and $t\equiv 2\pmod{4}$,
\item[(ii)] $u=5$ and $t\in \{17,19,23,29,31\}$.
\end{itemize}
\end{lemma}

\begin{proposition}
\label{StoG} If there exists a skew Room frame of type $t^u$,
then there exists a $[2,2]$-GDC$(6)$ of type $(6t)^u$ with size $6t^2u(u-1)$.
\end{proposition}
\begin{proof} Let $F$ be a given skew Room frame of
type $t^u$.  We construct a $[2,2]$-GDC$(6)$ of type $(6t)^u$
on group set $\{\{(i+k,j): 0\leq i\leq t-1, j\in \bbZ_6\}:
k=0,t,\dots,t(u-1)\}$. The code contains all the
codewords $\langle(a,j),(b,j),(c,1+j),(r,4+j)\rangle$,
$\langle(c,4+j),(r,1+j),(a,j),(b,j)\rangle$, where
$j\in \bbZ_6$ and the pair $\{a,b\}$ is contained in column $c$ and row $r$ of $F$.
It is easy to check that this code is of distance $6$ and composition $[2, 2]$.
\end{proof}

Combining  Lemma \ref{SRF} and Proposition \ref{StoG}, we have
the following result.

\begin{lemma}
\label{SRF2GDC(6t)^u} Let $u\geq 4$ and $t(u-1)$ be even. Then there
exists a $[2,2]$-GDC$(6)$ of type $(6t)^u$ with size $6t^2u(u-1)$,
except possibly for:
\begin{itemize}
\item[(i)] $u=4$ and $t\equiv 2\pmod{4}$,
\item[(ii)] $u=5$ and $t\in \{1,17,19,23,29,31\}$.
\end{itemize}
\end{lemma}

\subsection{Difference Matrix Construction}

Let $G$  be an abelian group of order $g$. A {\em difference matrix} based on $G$, denoted $(g,k;1)$-DM, is a $k\times g$ matrix $M=[m_{i,j}]$, $m_{i,j}$ in $G$, such that for each $1\leq r < s \leq k$, the differences $m_{r,j}-m_{s,j}, 1\leq j\leq g$, comprise all the elements of $G$.

\begin{theorem} [\cite{Ge:2005}]
\label{DM}
A $(g,4;1)$-DM exists if and only if $g\geq 4$ and $g\not \equiv 2\pmod{4}$.
\end{theorem}

\begin{proposition}
\label{DMtoGDC} If there exists a $(g,4;1)$-DM over $\bbZ_{g}$,
then there exists a $[2,2]$-GDC$(6)$ of type $g^4$ with size $2g^2$.
\end{proposition}
\begin{proof} Let $M=[m_{i,j}]$ be a given $(g,4;1)$-DM over $\bbZ_{g}$.  We construct a $[2,2]$-GDC$(6)$ of type $g^4$ on group set $\{\{(i,j): j \in \bbZ_{g}\}: 0\leq i\leq 3\}$. The code contains all the
codewords $\langle(0,m_{1,j}+k),(1,m_{2,j}+k),(2,m_{3,j}+k),(3,m_{4,j}+k)\rangle$,
$\langle(2,m_{3,j}+k),(3,m_{4,j}+k),(0,m_{1,j}+1+k),(1,m_{2,j}+1+k)\rangle$, where
$j,k\in \bbZ_g$.
It is easy to check that this code is of distance $6$ and composition $[2, 2]$.
\end{proof}

Combining  Theorem~\ref{DM} and Proposition~\ref{DMtoGDC}, we have
the following result.

\begin{lemma}
\label{DM2GDCg^4} There exists a $[2,2]$-GDC$(6)$ of type $g^4$ with size $2g^2$ for every $g\geq 4$ and $g\not \equiv 2\pmod{4}$.
\end{lemma}

\subsection{Some $[2,2]$-GDC$(6)$s}

\begin{lemma}
\label{[2,2]-GDC:2^10} There exists a $[2,2]$-GDC$(6)$ of type $2^{10}$ with size $60$.
\end{lemma}

\begin{proof}
Let $X=\bbZ_{20}$, and ${\cal G}=\{\{i,i+10\}:0\leq i\leq 9\}$. Then $(X,{\cal G},{\cal C})$ is a $[2,2]$-GDC$(6)$ of type $2^{10}$, where  ${\cal C}$ is obtained by developing the elements of $\bbZ_{20}$ in the codewords $\langle0,5,3,7\rangle$, $\langle0,4,1,13\rangle$,  $\langle0,8,14,19\rangle$ $+1\pmod{20}$.
\end{proof}

\begin{lemma}
\label{[2,2]-GDC:6^t} There exists a $[2,2]$-GDC$(6)$ of type $6^t$ with
size $6t(t-1)$ for each $5\leq t\leq 11$.
\end{lemma}

\begin{proof} For $t\in \{5,8\}$, let $X_t=\bbZ_{6t}$, and ${\cal
G}_t=\{\{i,i+t,i+2t,i+3t,i+4t,i+5t\}:0\leq i\leq t-1\}$. Then
$(X_t,{\cal G}_t,{\cal C}_t)$ is a $[2,2]$-GDC$(6)$ of type
$6^t$ and size $6t(t-1)$, where ${\cal C}_t$ is obtained by
developing the elements of $\bbZ_{6t}$ in the following
codewords $+1\pmod{6t}$.

\noindent $t=5$: $\langle0,24,1,13\rangle$
$\langle0,9,8,11\rangle$ $\langle0,3,17,26\rangle$
$\langle0,12,4,28\rangle$

\noindent $t=8$: $\langle0,18,13,3\rangle$
$\langle0,28,21,25\rangle$ $\langle0,10,22,36\rangle$
$\langle0,46,27,9\rangle$ $\langle0,14,31,37\rangle$
$\langle0,6,5,7\rangle$ $\langle0,4,19,39\rangle$

For $t=6$, let $X_6=\bbZ_{12}\times I_3$, and ${\cal
G}_6=\{\{(i,0),(i+6,0),(i,1),(i+6,1),(i,2),(i+6,2)\}:0\leq
i\leq 5\}$. Then $(X_6,{\cal G}_6,{\cal C}_6)$ is a
$[2,2]$-GDC$(6)$ of type $6^6$ and size $180$, where ${\cal
C}_6$ is obtained by developing the elements of $\bbZ_{12}\times
I_3$ in the following codewords $(+1\pmod{12},-)$.
$$\begin{array}{lllll}
\langle0_1, 9_1, 10_0, 5_0\rangle &
\langle0_0, 5_0, 8_0, 9_0\rangle &
\langle0_0, 2_0, 7_2, 4_2\rangle &
\langle0_0, 11_0, 7_1, 10_1\rangle&
\langle0_2, 2_2, 4_1, 7_0\rangle \\
\langle0_2, 9_2, 1_0, 11_0\rangle &
\langle0_1, 3_0, 4_2, 2_2\rangle &
\langle0_2, 5_2, 8_0, 1_1\rangle &
\langle0_1, 11_0, 8_2, 7_2\rangle &
\langle0_1, 2_1, 4_0, 7_1\rangle \\
\langle0_1, 7_0, 10_2, 5_2\rangle &
\langle0_2, 9_1, 5_1, 8_2\rangle  &
\langle0_1, 9_2, 1_2, 4_1\rangle &
\langle0_1, 9_0, 11_1, 1_1\rangle &
\langle0_2, 11_2, 10_1, 9_0\rangle \\
\end{array}$$

For $t\in\{7,9,11\}$, the desired GDCs are obtained from Lemma~\ref{SRF2GDC(6t)^u}. For $t=10$, inflate a $[2,2]$-GDC$(6)$ of type $2^{10}$ with weight $3$ to obtain the desired code.
\end{proof}

\begin{lemma}
\label{[2,2]-GDC:sg^u} There exists a $[2,2]$-GDC$(6)$ of type $g^u$ with size $u(u-1)g^2/6$ for each $(g,u) \in \{(3,7), (3,11), (3,13),$ $(10,7), (18,4)\}$.
\end{lemma}

\begin{proof} For each $[2,2]$-GDC$(6)$ of type $g^u$, let $X=\bbZ_{gu}$, and ${\cal
G}=\{\{i,i+u,i+2u,\dots,i+(g-1)u\}:0\leq i\leq u-1\}$. Then $(X,{\cal
G},{\cal C})$ is a $[2,2]$-GDC$(6)$ of type $g^u$, where ${\cal C}$ is obtained by developing the elements
of $\bbZ_{gu}$ in the following codewords under the automorphism group as below.

\noindent $3^7$: $+1\pmod{21}$
$$\begin{array}{lll}
\langle7,3,18,13\rangle & \langle4,3,12,16\rangle & \langle0,5,2,3\rangle
\end{array}$$

\noindent $3^{11}$: $+1\pmod{33}$
$$\begin{array}{lllll}
 \langle25,15,0,7\rangle & \langle13,27,15,25\rangle & \langle0,7,1,6\rangle &
 \langle10,5,1,14\rangle & \langle2,15,18,32\rangle \\
\end{array}$$

\noindent $3^{13}$: $+1\pmod{39}$
$$\begin{array}{llllll}
\langle20,2,25,10\rangle & \langle21,6,23,17\rangle  & \langle6,12,5,26\rangle &
\langle30,3,28,19\rangle & \langle21,24,31,4\rangle & \langle0,9,1,4\rangle
\end{array}$$


\noindent $10^7$: $+1\pmod{70}$
$$\begin{array}{llllll}
\langle0,31,50,5\rangle & \langle0,61,24,46\rangle & \langle0,25,6,68\rangle & \langle0,32,16,52\rangle &
\langle0,3,2,40\rangle & \langle0,34,47,64\rangle \\
\langle0,48,1,60\rangle & \langle0,17,27,58\rangle & \langle0,11,26,29\rangle &  \langle0,62,57,66\rangle \\
\end{array}$$

\noindent $18^4$: $+2\pmod{72}$
$$\begin{array}{llllll}
\langle1,3,8,22\rangle &
\langle0,34,1,15\rangle &
\langle1,7,16,54\rangle &
\langle0,2,21,71\rangle &
\langle0,9,54,63\rangle &
\langle0,10,5,47\rangle \\
\langle0,6,51,41\rangle &
\langle0,7,18,25\rangle &
\langle43,5,2,28\rangle &
\langle43,1,4,14\rangle &
\langle0,26,11,13\rangle &
\langle1,15,50,56\rangle \\
\langle0,30,23,61\rangle &
\langle0,14,43,17\rangle &
\langle0,22,49,55\rangle &
\langle1,23,60,18\rangle &
\langle1,27,26,28\rangle &
\langle1,11,62,40\rangle \\
\end{array}$$
\end{proof}

\begin{lemma}
\label{[2,2]-GDC:s6^t3^1} There exists a $[2,2]$-GDC$(6)$ of type
$6^t3^1$ with size $6t^2$ for each $t=4$ or $6\leq t\leq 11$.
\end{lemma}

\begin{proof} Let $X_t=I_{6t+3}$, and ${\cal
G}_t=\{\{i,i+t,i+2t,\dots,i+5t\}:0\leq i\leq
t-1\}\cup\{\{6t,6t+1,6t+2\}\}$. Then $(X_t,{\cal G}_t,{\cal
C}_t)$ is a $[2,2]$-GDC$(6)$ of type $6^t3^1$, where ${\cal
C}_4$ is obtained by developing
the following codewords under the
automorphism group $G=\langle(0\ \ 2\ \ 4\ \ \cdots\ \ 22)(1\ \
3\ \ 5\ \ \cdots\ \ 23)$ $(24\ \ 25\ \ 26)\rangle$, and ${\cal C}_t$ for $6\leq t\leq 11$ is obtained by developing
the following codewords under the
automorphism group $G=\langle(0\ \ 1\ \ 2\ \ \cdots\ \
6t-1)(6t\ \ 6t+1\ \ 6t+2)\rangle$.

\noindent $t=4$:
$$\begin{array}{llll}
\langle1,25,8,6\rangle & \langle0,10,24,23\rangle & \langle1,7,10,20\rangle & \langle0,26,19,17\rangle \\ \langle0,22,7,1\rangle & \langle1,11,12,26\rangle & \langle0,6,21,11\rangle & \langle1,3,18,0\rangle \\
\end{array}$$

\noindent $t=6$:
$$\begin{array}{llllll}
\langle0,34,27,5\rangle & \langle0,28,3,13\rangle & \langle0,10,19,35\rangle & \langle0,20,17,15\rangle &
\langle0,37,32,4\rangle & \langle0,14,1,38\rangle \\
\end{array}$$

\noindent $t=7$:
$$\begin{array}{llllll}
\langle0,26,25,44\rangle & \langle0,8,12,38\rangle & \langle0,43,20,22\rangle & \langle0,40,1,9\rangle &
\langle0,6,23,33\rangle & \langle0,18,31,37\rangle \\ \langle0,10,15,39\rangle \\
\end{array}$$

\noindent $t=8$:
$$\begin{array}{llllll}
\langle0,10,12,46\rangle & \langle0,28,5,23\rangle & \langle0,26,41,45\rangle & \langle0,18,3,9\rangle &
\langle0,14,13,49\rangle & \langle0,42,1,29\rangle \\ \langle0,48,11,37\rangle & \langle0,4,21,31\rangle\\
\end{array}$$

\noindent $t=9$:
$$\begin{array}{llllll}
\langle0,52,5,21\rangle & \langle0,54,49,53\rangle & \langle0,38,3,51\rangle & \langle0,14,1,47\rangle &
\langle0,8,43,55\rangle & \langle0,34,12,22\rangle \\ \langle0,6,17,37\rangle & \langle0,44,15,29\rangle &
\langle0,4,28,30\rangle  \\
\end{array}$$

\noindent $t=10$:
$$\begin{array}{llllll}
\langle0,12,36,37\rangle & \langle0,54,45,57\rangle & \langle0,8,29,35\rangle & \langle0,34,23,61\rangle &
\langle0,62,19,41\rangle & \langle0,4,9,17\rangle \\ \langle0,1,15,47\rangle & \langle0,58,16,42\rangle &
\langle0,28,7,11\rangle & \langle0,22,53,55\rangle
\end{array}$$

\noindent $t=11$:
$$\begin{array}{llllll}
\langle0,10,1,13\rangle & \langle0,14,43,68\rangle & \langle0,12,30,38\rangle & \langle0,6,15,47\rangle &
\langle0,8,59,61\rangle & \langle0,21,49,63\rangle\\ \langle0,67,7,17\rangle  & \langle0,16,36,40\rangle &
\langle0,64,25,46\rangle & \langle0,34,5,65\rangle & \langle0,4,23,39\rangle   \\
\end{array}$$
\end{proof}

\begin{lemma}
\label{[2,2]-GDC:s6^t9^1} There exists a $[2,2]$-GDC$(6)$ of type
$6^t9^1$ with size $6t(t+2)$ for each $t\in \{6,7,8\}$.
\end{lemma}

\begin{proof} Let $X_t=I_{6t+9}$, and ${\cal
G}_t=\{\{i,i+t,i+2t,\dots,i+5t\}:0\leq i\leq
t-1\}\cup\{\{6t,6t+1,6t+2,\dots,6t+8\}\}$. Then $(X_t,{\cal
G}_t,{\cal C}_t)$ is a $[2,2]$-GDC$(6)$ of type $6^t9^1$, where
${\cal C}_t$ is obtained by developing the following codewords under the
automorphism group $G=\langle(0\ \ 1\ \ 2\ \ \cdots\ \
6t-1)(6t\ \ 6t+1\ \ 6t+2\ \ 6t+3\ \ 6t+4\ \ 6t+5)$ $(6t+6\ \
6t+7\ \ 6t+8)\rangle$.

\noindent $t=6$:
$$\begin{array}{llll}
\langle0,36,5,15\rangle & \langle0,22,13,29\rangle & \langle0,34,21,38\rangle & \langle24,4,7,37\rangle\\
\langle0,10,9,11\rangle & \langle0,43,17,31\rangle & \langle0,41,33,25\rangle & \langle0,8,4,44\rangle \\
\end{array}$$

\noindent $t=7$:
$$\begin{array}{llll}
\langle0,36,34,10\rangle & \langle0,45,8,39\rangle & \langle0,23,25,49\rangle & \langle0,44,3,22\rangle  \\
\langle0,24,41,12\rangle & \langle0,5,46,32\rangle & \langle0,29,33,42\rangle & \langle0,48,38,1\rangle \\
\langle0,11,20,26\rangle \\
\end{array}$$

\noindent $t=8$:
$$\begin{array}{llll}
\langle0,10,3,15\rangle & \langle0,14,54,37\rangle & \langle0,49,27,29\rangle & \langle0,36,31,1\rangle  \\
\langle0,18,9,35\rangle & \langle0,56,47,19\rangle & \langle0,22,51,33\rangle & \langle0,46,42,4\rangle \\
\langle0,48,7,21\rangle & \langle0,20,52,45\rangle\\
\end{array}$$

\end{proof}

\begin{lemma}
\label{[2,2]-GDC:sg^um^1} There exists a $[2,2]$-GDC$(6)$ of type
$g^u m^1$ with size $(g^2u(u-1)+2gum)/6$ for each $(g,u,m) \in
\{(9,5,9), (9,5,15),$ $ (18,4,6), (18,4,12), (18, 6, 33),(24,4,6),$
$(24,4,9)\}$.
\end{lemma}

\begin{proof}
For each $[2,2]$-GDC$(6)$ of type $g^u m^1$, let $X=I_{gu+m}$, and ${\cal
G}=\{\{0,u,2u,\dots,(g-1)u\}+i:0\leq i\leq
u-1\}\cup\{\{gu,gu+1,gu+2,\dots,gu+m-1\}\}$. Then $(X,{\cal
G},{\cal C})$ is the desired $[2,2]$-GDC$(6)$, where
${\cal C}$ is obtained by developing the following codewords under the
automorphism group as below.

\noindent $9^5 9^1$: $G=\langle(0\ \ 1\ \ 2\ \ \cdots\ \
44)(45\ \ 46\ \ \cdots\ \ 53)\rangle$
$$\begin{array}{llllll}
\langle2, 45, 43, 1\rangle &
\langle0, 9, 17, 28\rangle &
\langle24, 21, 42, 3\rangle &
\langle4, 45, 30, 8\rangle &
\langle8, 42, 45, 40\rangle &
\langle17, 40, 24, 33\rangle \\
\langle1, 7, 45, 38\rangle &
\langle0, 12, 45, 14\rangle &
\langle14, 45, 27, 15\rangle \\
\end{array}$$

\noindent $9^5 15^1$: $G=\langle(0\ \ 1\ \ 2\ \ \cdots\ \
44)(45\ \ 46\ \ \cdots\ \ 59)\rangle$
$$\begin{array}{llllll}
\langle45, 35, 18, 27\rangle &
\langle45, 13, 9, 16\rangle &
\langle45, 41, 4, 17\rangle &
\langle24, 15, 41, 45\rangle &
\langle45, 22, 40, 6\rangle &
\langle0, 11, 42, 44\rangle \\
\langle18, 16, 22, 45\rangle &
\langle8, 21, 35, 45\rangle &
\langle4, 27, 43, 45\rangle &
\langle10, 17, 29, 45\rangle &
\langle45, 14, 15, 38\rangle &
\end{array}$$

\noindent $18^4 6^1$: $G=\langle(0\ \ 1\ \ 2\ \ \cdots\ \
71)(72\ \ 73\ \ 74\ \ 75\ \ 76\ \ 77)\rangle$
$$\begin{array}{llllll}
\langle0,26,9,19\rangle & \langle0,54,21,59\rangle & \langle0,50,53,76\rangle & \langle0,72,29,51\rangle &
\langle0,2,13,71\rangle & \langle0,30,1,7\rangle \\ \langle0,34,61,73\rangle & \langle0,66,17,35\rangle &
\langle0,10,25,67\rangle & \langle0,58,31,33\rangle & \langle0,77,37,63\rangle \\
\end{array}$$

\noindent $18^4 12^1$: $G=\langle(0\ \ 1\ \ 2\ \ \cdots\ \
71)(72\ \ 73\ \ 74\ \ 75\ \ 76\ \ 77)$ $(78\ \ 79\ \ 80\ \ 81\ \ 82\ \ 83)\rangle$
$$\begin{array}{llllll}
\langle0,46,7,72\rangle & \langle0,77,57,71\rangle & \langle0,34,63,79\rangle & \langle0,22,49,78\rangle &
\langle0,14,1,31\rangle & \langle0,10,47,65\rangle \\ \langle0,83,9,11\rangle & \langle0,18,23,61\rangle &
\langle0,42,39,45\rangle & \langle0,82,21,67\rangle & \langle0,6,19,41\rangle & \langle0,70,51,73\rangle \\
\langle0,76,15,25\rangle\\
\end{array}$$

\noindent $18^6 33^1$: $G=\langle(0\ \ 1\ \ 2\ \ \cdots\ \ 107)(108\
\ 109\ \ \cdots\ \  113)$ $(114\ \ 115\ \ \cdots\ \  119)(120\ \
121\ \ \cdots\ \  125)$ $(126\ \ 127\ \ \cdots\ \  131)$ $(132\ \ 133\
\ \cdots \ \ 140)\rangle$
$$\begin{array}{llllll}
\langle0,138,53,97\rangle & \langle0,20,17,131\rangle & \langle0,82,7,128\rangle &
\langle0,70,92,139\rangle & \langle0,52,27,108\rangle & \langle0,120,4,9\rangle \\
\langle0,121,68,81\rangle & \langle0,95,58,140\rangle & \langle0,47,43,63\rangle &
\langle0,74,67,118\rangle & \langle0,73,51,124\rangle & \langle0,8,40,133\rangle \\
\langle0,115,37,99\rangle & \langle0,136,28,29\rangle & \langle0,39,25,98\rangle &
\langle0,117,10,31\rangle & \langle0,127,11,80\rangle & \langle0,111,19,93\rangle \\
\langle0,87,89,125\rangle & \langle0,137,23,49\rangle & \langle0,44,15,85\rangle &
\langle0,62,45,109\rangle & \langle0,126,14,75\rangle & \langle0,112,57,65\rangle \\
\langle0,107,76,119\rangle & \langle0,103,50,106\rangle &
\end{array}$$

\noindent $24^4 6^1$: $G=\langle(0\ \ 1\ \ 2\ \ \cdots\ \
95)(96\ \ 97\ \ 98\ \ 99\ \ 100\ \ 101)\rangle$
$$\begin{array}{llllll}
\langle0,86,25,75\rangle & \langle0,82,81,97\rangle  & \langle0,42,23,29\rangle & \langle0,101,63,65\rangle  &
\langle0,30,21,47\rangle & \langle0,62,5,19\rangle\\ \langle0,38,11,89\rangle & \langle0,26,3,98\rangle &
\langle0,18,27,49\rangle & \langle0,100,57,67\rangle & \langle0,74,37,71\rangle & \langle0,2,15,45\rangle \\
\langle0,46,41,79\rangle & \langle0,6,7,61\rangle\\
\end{array}$$

\noindent $24^4 9^1$: $G=\langle(0\ \ 1\ \ 2\ \ \cdots\ \
95)(96\ \ 97\ \ 98)(99\ \ 100\ \ 101)$ $(102\ \ 103\ \ 104)\rangle$
$$\begin{array}{llllll}
\langle0,82,15,73\rangle & \langle0,38,3,49\rangle & \langle0,34,23,101\rangle & \langle0,10,81,7\rangle &
\langle0,26,19,25\rangle & \langle0,54,41,59\rangle \\ \langle0,18,45,31\rangle & \langle0,6,39,69\rangle &
\langle0,22,17,104\rangle & \langle0,30,51,9\rangle & \langle0,96,43,77\rangle & \langle0,46,47,98\rangle\\
\langle0,99,53,79\rangle & \langle0,2,57,67\rangle  & \langle0,102,35,37\rangle \\
\end{array}$$
\end{proof}

Let $P=[9,19]\cup[21,23]\cup[26,28]\cup[31,33]$.

\begin{lemma}
\label{[2,2]-GDCa} For each $t\geq 9$ and $t\not\in P$, there exists
a $[2,2]$-GDC$(6)$ of type $24^i30^j36^k42^l48^m$, where $i$,
$j$, $k$, $l$, $m$ are integers such that $4i+5j+6k+7l+8m=t$.
\end{lemma}

\begin{proof} For each $t\geq 9$ and $t\not\in P$, take
a $(t+1,\{5,6,7,8,9\},1)$-PBD from Theorem~\ref{PBD5-9}, and
remove one point from the point set to obtain a
$\{5,6,7,8,9\}$-GDD of type $4^i 5^j 6^k 7^l 8^m$ with
$4i+5j+6k+7l+8m=t$. Apply the Fundamental Construction with
weight $6$ to this GDD and input $[2,2]$-GDC$(6)$s of type
$6^u$ for $u\in \{5,6,7,8,9\}$ from Lemma~\ref{[2,2]-GDC:6^t} to obtain
a $[2,2]$-GDC$(6)$ of type $24^i 30^j 36^k 42^l 48^m$.
\end{proof}

\begin{lemma}
\label{[2,2]-GDCb} The following $[2,2]$-GDC$(6)$s all exist:
\begin{enumerate}
\item[i)] type $18^u$ and size $54u(u-1)$ for $u\in
    \{4,5,6,7,$ $9,11\}$;
\item[ii)] type $24^u$ and size $96u(u-1)$ for $u\in
    \{4,7,8\}$;
\item[iii)] type $24^u 36^1$ and size $96u(u+2)$ for $u\in
    \{4,5\}$;
\item[iv)] type $18^8 42^1$ and size $5040$;
\item[v)] type $30^4 18^1$ and size $2520$;
\item[vi)] type $24^4 18^1$ and size $1728$.
\end{enumerate}
\end{lemma}

\begin{proof}
For i), a $[2,2]$-GDC$(6)$ of type $18^4$ is constructed
directly in Lemma~\ref{[2,2]-GDC:sg^u}. For each $t\in \{5,6,7,9,11\}$, take a
$[2,2]$-GDC$(6)$ of type $6^t$ from Lemma~\ref{[2,2]-GDC:6^t}, and
inflate it with weight $3$ to obtain a $[2,2]$-GDC$(6)$ of type
$18^t$.

For ii), the required GDCs are obtained from
Lemma~\ref{SRF2GDC(6t)^u}.

For iii), the required GDCs are obtained by applying the
Fundamental Construction with weight $4$ to a $4$-GDD of
type $6^u 9^1$ (see \cite[Theorem 1.6]{GeRees:2004}). The input
$[2,2]$-GDC$(6)$s of type $4^4$ come from Lemma~\ref{DM2GDCg^4}.

For iv), take a $5$-GDD of type $3^8 7^1$ obtained by
completing the parallel classes of a $4$-RGDD of type $3^8$
(see \cite{Shen:1992}). Apply the Fundamental Construction with weight
$6$ to obtain a $[2,2]$-GDC$(6)$ of type $18^8 42^1$. The input
$[2,2]$-GDC$(6)$s of type $6^5$ come from Lemma~\ref{[2,2]-GDC:6^t}.

For v), take a TD$(5,5)$ and apply the Fundamental Construction, giving weight $6$ to each point in the first four groups and
one point in the last group and weight 3 to each of the remaining
points. Noting that there exist $[2,2]$-GDC$(6)$s of types $6^5$ and $6^4 3^1$ by Lemmas~\ref{[2,2]-GDC:6^t} and \ref{[2,2]-GDC:s6^t3^1},
we get a $[2,2]$-GDC$(6)$ of type $30^4 18^1$.

For vi), take a TD$(5,4)$ and apply the Fundamental Construction, giving weight $6$ to
each point in the first four groups and two points in the last
group and weight $3$ to each of the remaining points. Noting that there exist
$[2,2]$-GDC$(6)$s of types $6^5$ and $6^4 3^1$ by Lemmas~\ref{[2,2]-GDC:6^t} and \ref{[2,2]-GDC:s6^t3^1},
we obtain a $[2,2]$-GDC$(6)$ of type $24^4 18^1$.
\end{proof}

\subsection{Cases of Length $n\equiv 0,1\pmod{6}$}

\begin{lemma}
$A_3(7,6,[2,2])=3$, $A_3(13,6,[2,2])\geq 21$.
\end{lemma}

\begin{proof} For $n=7$, see Table~\ref{presult}.

For $n=13$, the required code is constructed on
$I_{13}$ with codewords as below.
$$\begin{array}{llllll}
\langle9,3,5,6\rangle & \langle11,7,1,6\rangle & \langle9,12,11,7\rangle & \langle2,5,0,7\rangle & \langle1,4,7,10\rangle & \langle9,10,0,1\rangle \\ \langle0,7,8,9\rangle & \langle3,12,1,4\rangle & \langle0,12,10,5\rangle & \langle0,1,2,3\rangle & \langle10,5,8,4\rangle & \langle4,11,2,9\rangle \\
\langle1,6,9,12\rangle & \langle0,6,4,11\rangle & \langle10,7,12,2\rangle & \langle8,6,3,7\rangle  & \langle1,8,11,5\rangle & \langle2,12,6,8\rangle \\ \langle8,4,0,12\rangle & \langle2,3,10,11\rangle &  \langle5,11,3,12\rangle\\
\end{array}$$
\end{proof}

\begin{lemma}
\label{s6t+1} $A_3(6t+1,6,[2,2])=U(6t+1,6,[2,2])$ for each
$t\in [3,11]\cup\{13,14,17\}$.
\end{lemma}

\begin{proof} For each $t\in [3,11]\cup\{13,14,17\}$ and $t\neq 4$, let $X_t=\bbZ_{6t+1}$. Then $(X_t,{\cal C}_t)$ is the
desired optimal $(6t+1,6,[2,2])_3$-code, where ${\cal C}_t$ is
obtained by developing the elements of $\bbZ_{6t+1}$ in the
codewords listed in Table~\ref{t6t+1} $+1\pmod{6t+1}$.

Let $X_4=\bbZ_5\times\bbZ_5$. Then $(X_4,{\cal C}_4)$ is the
desired optimal $(25,6,[2,2])_3$-code, where ${\cal C}_4$ is
obtained by developing the elements of $\bbZ_5\times\bbZ_5$ in
the codewords listed in Table~\ref{t6t+1}
$(+1\pmod{5},+1\pmod{5})$.
\end{proof}

\begin{table*}
\scriptsize
\centering
\renewcommand{\arraystretch}{1}
\setlength\arraycolsep{3pt} \caption{Base Codewords of Small
Optimal $(6t+1,6,[2,2])_3$-Codes in Lemma~\ref{s6t+1}}
\label{t6t+1}
\begin{tabular}{c|l}
\hline $t$ & {\hfill Codewords \hfill}
\\
\hline $3$ & $\begin{array}{llllll}
\langle0,1,4,16\rangle & \langle0,7,9,17\rangle & \langle0,8,13,14\rangle &
\end{array}$
\\
\hline $4$ & $\begin{array}{llllll}
\langle0_3, 2_0, 3_0, 4_3\rangle & \langle0_2, 2_3, 2_1, 0_4\rangle & \langle0_2, 1_0, 1_4, 0_3\rangle &
\langle0_0, 1_1, 2_0, 4_1\rangle
\end{array}$
\\
\hline $5$ & $\begin{array}{llllll}
\langle0,3,11,23\rangle & \langle0,7,2,5\rangle & \langle0,6,15,22\rangle & \langle0,14,4,10\rangle & \langle0,12,13,30\rangle &
\end{array}$
\\
\hline $6$ & $\begin{array}{lllllll}
\langle0,23,27,35\rangle & \langle0,8,11,17\rangle & \langle0,5,25,1\rangle & \langle0,6,22,36\rangle & \langle0,18,2,7\rangle & \langle0,13,10,28\rangle &
\end{array}$
\\
\hline $7$ & $\begin{array}{lllllll}
\langle0,35,26,23\rangle & \langle0,16,33,37\rangle & \langle0,29,22,38\rangle & \langle0,3,10,18\rangle & \langle0,30,11,12\rangle & \langle0,1,6,20\rangle & \langle0,4,2,32\rangle \\
\end{array}$
\\
\hline $8$ & $\begin{array}{llllllll}
\langle0,36,40,45\rangle & \langle0,28,14,47\rangle & \langle0,42,10,32\rangle & \langle0,33,25,18\rangle & \langle0,44,15,26\rangle & \langle0,11,48,12\rangle & \langle0,43,23,2\rangle & \langle0,22,3,46\rangle \\
\end{array}$
\\
\hline $9$ & $\begin{array}{lllllllll}
\langle0,8,1,21\rangle & \langle0,43,19,42\rangle & \langle0,32,34,39\rangle & \langle0,20,30,45\rangle & \langle0,28,9,26\rangle & \langle0,17,41,14\rangle & \langle0,15,6,18\rangle & \langle0,5,16,49\rangle & \\ \langle0,22,4,51\rangle \\
\end{array}$
\\
\hline $10$ & $\begin{array}{lllllllll}
\langle0,27,36,41\rangle & \langle0,11,21,39\rangle & \langle0,3,22,26\rangle & \langle0,1,47,53\rangle & \langle0,5,35,37\rangle & \langle0,17,24,25\rangle & \langle0,2,15,42\rangle &
\langle0,4,33,16\rangle &\\ \langle0,6,51,54\rangle & \langle0,18,38,49\rangle &
\end{array}$
\\
\hline $11$ & $\begin{array}{lllllllll}
\langle0,20,33,44\rangle & \langle0,1,32,41\rangle & \langle0,10,49,53\rangle & \langle0,30,48,51\rangle & \langle0,3,28,65\rangle & \langle0,7,26,36\rangle & \langle0,11,16,17\rangle &
\langle0,8,23,35\rangle &\\ \langle0,9,54,61\rangle & \langle0,4,42,50\rangle & \langle0,12,14,34\rangle &
\end{array}$
\\
\hline $13$ & $\begin{array}{lllllllll}
\langle0,11,54,56\rangle & \langle0,5,49,60\rangle & \langle0,18,35,64\rangle & \langle0,63,66,76\rangle & \langle0,29,65,71\rangle & \langle0,10,24,33\rangle & \langle0,20,51,67\rangle &
\langle0,7,15,19\rangle &\\ \langle0,1,22,40\rangle & \langle0,9,57,62\rangle & \langle0,4,34,41\rangle & \langle0,2,27,28\rangle & \langle0,6,38,58\rangle &
\end{array}$
\\
\hline $14$ & $\begin{array}{lllllllll}
\langle0,9,41,53\rangle & \langle0,1,50,83\rangle & \langle0,45,47,63\rangle & \langle0,46,61,70\rangle & \langle0,4,25,59\rangle & \langle0,17,27,28\rangle & \langle0,52,57,65\rangle &
\langle0,8,56,75\rangle &\\ \langle0,16,74,80\rangle & \langle0,19,22,62\rangle & \langle0,6,29,36\rangle & \langle0,7,38,42\rangle & \langle0,12,26,72\rangle & \langle0,34,54,71\rangle \\
\end{array}$
\\
\hline $17$ & $\begin{array}{lllllllll}
\langle0,20,62,63\rangle & \langle0,10,58,88\rangle & \langle0,3,59,95\rangle & \langle0,17,24,50\rangle & \langle0,1,76,97\rangle & \langle0,12,46,66\rangle & \langle0,26,39,41\rangle & \langle0,14,69,79\rangle & \\
\langle0,36,64,81\rangle & \langle0,74,85,99\rangle & \langle0,5,40,49\rangle & \langle0,19,87,90\rangle & \langle0,9,47,70\rangle & \langle0,21,53,72\rangle & \langle0,30,52,57\rangle & \langle0,23,31,60\rangle &  \\
\langle0,2,6,18\rangle \\
\end{array}$
\\
\hline
\end{tabular}
\end{table*}

\begin{lemma}
$A_3(6t+1,6,[2,2])=U(6t+1,6,[2,2])$ for each $t\geq 12$ and
$t\not\in\{13,14,17\}$.
\end{lemma}

\begin{proof} For each $t\geq 12$ and $t\not\in P$, take a
$[2,2]$-GDC$(6)$ of type $24^i30^j36^k42^l48^m$ with
$4i+5j+6k+7l+8m=t$ from Lemma~\ref{[2,2]-GDCa}. Adjoin one
ideal point, and fill in the groups together with the ideal
point with optimal codes of small lengths from Lemma \ref{s6t+1} to obtain the desired
code.

For each $t\in \{12,15,16,18,19,21,22,23,26,27,28,31,32,33\}$, take
a $[2,2]$-GDC$(6)$ constructed in Lemma~\ref{[2,2]-GDCb}. Adjoin
one ideal point, and fill in the groups together with the ideal
point with optimal codes of small lengths from Lemma \ref{s6t+1} to obtain the desired
code.
\end{proof}

Combining the above lemmas, we obtain the following result for $n\equiv 1\pmod{6}$.

\begin{theorem}
\label{6t+1} $A_3(7,6,[2,2])=3$, $A_3(13,6,[2,2])\geq 21$,
$A_3(6t+1,6,[2,2])=U(6t+1,6,[2,2])$ for each $t\geq 3$;
\end{theorem}

For $n\equiv 0\pmod{6}$, we have the following result.

\begin{theorem}
\label{6t} $A_3(6t,6,[2,2])=U(6t,6,[2,2])$ for each $t\geq 1$.
\end{theorem}

\begin{proof} For $t=1$, see Table~\ref{presult}.
For $t=2$, let $X=\bbZ_{12}$. Then the code ${\cal C}$ is obtained
by developing the elements of $\bbZ_{12}$ in the following
codewords $+2\pmod{12}$.
$$\begin{array}{lll}
\langle0,2,8,5\rangle &
\langle0,9,7,11\rangle &
\langle1,5,0,10\rangle
\end{array}$$

For each $t\geq 3$, remove one point and related codewords from an optimal $(6t+1,6,[2,2])_3$-code from Theorem~\ref{6t+1} to get the desired code.
\end{proof}

\subsection{Case of Length  $n\equiv 2\pmod{6}$}

\begin{lemma}
$A_3(8,6,[2,2])=5$, $A_3(14,6,[2,2])\geq 27$.
\end{lemma}

\begin{proof}
For $n=8$, see Table~\ref{presult}.

For $n=14$, the required code is constructed on
$I_{14}$ with codewords as below.
$$\begin{array}{llllll}
\langle0,2,5,9\rangle & \langle5,9,6,10\rangle & \langle6,10,3,4\rangle & \langle0,7,11,12\rangle &
\langle8,13,0,2\rangle & \langle5,7,1,13\rangle \\ \langle7,9,2,4\rangle & \langle4,11,5,7\rangle &
\langle3,5,2,11\rangle & \langle2,4,10,13\rangle & \langle9,11,0,3\rangle & \langle9,12,8,13\rangle \\
\langle1,8,3,5\rangle & \langle1,10,0,7\rangle & \langle6,13,7,9\rangle & \langle1,13,4,11\rangle &
\langle3,4,9,12\rangle & \langle10,13,5,12\rangle \\ \langle0,4,1,8\rangle & \langle2,12,3,7\rangle &
\langle5,12,0,4\rangle & \langle8,10,9,11\rangle & \langle1,6,2,12\rangle & \langle11,12,1,10\rangle \\
\langle3,7,8,10\rangle & \langle0,3,6,13\rangle &  \langle2,11,6,8\rangle\\
\end{array}$$
\end{proof}

It is easy to see that if there exists a $[2,2]$-GDC$(6)$ of
type $2^{3t+1}$ with size $2t(3t+1)$, then there is an optimal
$(6t+2,6,[2,2])_3$-code. So we will construct
$[2,2]$-GDC$(6)$s of type $2^{3t+1}$.

\begin{lemma}
\label{s2^{3t+1}} There exists a $[2,2]$-GDC$(6)$ of type
$2^{3t+1}$ with size $2t(3t+1)$ for each $3\leq t\leq 11$ or
$t\in\{14,17\}$.
\end{lemma}

\begin{proof} For $t=3$, the code is constructed in Lemma~\ref{[2,2]-GDC:2^10}.

For $4\leq t\leq 11$ or
$t\in\{14,17\}$, let $X_t=\bbZ_{6t+2}$, and ${\cal
G}_t=\{\{i,i+3t+1\}:0\leq i\leq 3t\}$. Then $(X_t,{\cal
G}_t,{\cal C}_t)$ is the desired $[2,2]$-GDC$(6)$ of type $2^{3t+1}$, where for $t\neq 5$, ${\cal C}_t$ is obtained
by developing the elements of $\bbZ_{6t+2}$ in the codewords
listed in Table~\ref{t2^{3t+1}} $+1\pmod{6t+2}$, and ${\cal
C}_5$ is obtained by developing the elements of $\bbZ_{32}$ in
the codewords listed in Table~\ref{t2^{3t+1}} $+2\pmod{32}$.
\end{proof}

\begin{table*}
\scriptsize
\centering
\renewcommand{\arraystretch}{1}
\setlength\arraycolsep{3pt} \caption{Base Codewords of Small
$[2,2]$-GDC$(6)$s of Type $2^{3t+1}$ in Lemma~\ref{s2^{3t+1}}}
\label{t2^{3t+1}}
\begin{tabular}{c|l}
\hline $t$ & {\hfill Codewords \hfill}
\\
\hline $4$ & $\begin{array}{llllll}
\langle0,4,3,11\rangle & \langle0,5,6,15\rangle & \langle0,9,2,23\rangle & \langle0,8,20,24\rangle &  &  \\
\end{array}$
\\
\hline $5$ & $\begin{array}{llllllll}
\langle0,28,9,29\rangle & \langle1,7,2,12\rangle & \langle0,15,24,7\rangle & \langle0,10,30,12\rangle & \langle0,18,23,21\rangle & \langle1,14,22,9\rangle & \langle1,21,4,8\rangle & \langle0,26,25,11\rangle \\
\langle1,15,11,5\rangle & \langle1,3,0,26\rangle &
\end{array}$
\\
\hline $6$ & $\begin{array}{llllll}
\langle0,7,20,5\rangle & \langle0,21,11,27\rangle & \langle0,23,14,35\rangle & \langle0,30,24,25\rangle & \langle0,22,18,26\rangle & \langle0,1,3,10\rangle \\
\end{array}$
\\
\hline $7$ & $\begin{array}{lllllll}
\langle0,5,21,33\rangle & \langle0,32,34,42\rangle & \langle0,14,23,29\rangle & \langle0,6,7,25\rangle & \langle0,36,40,35\rangle & \langle0,20,17,3\rangle & \langle0,18,11,31\rangle \\
\end{array}$
\\
\hline $8$ & $\begin{array}{llllllll}
\langle0,40,29,34\rangle & \langle0,30,2,11\rangle & \langle0,37,3,33\rangle & \langle0,45,12,27\rangle & \langle0,35,21,8\rangle & \langle0,49,18,42\rangle & \langle0,9,7,6\rangle & \langle0,24,28,38\rangle \\
\end{array}$
\\
\hline $9$ & $\begin{array}{lllllllll}
\langle0,39,33,45\rangle & \langle0,12,48,53\rangle & \langle0,22,35,9\rangle & \langle0,5,3,37\rangle & \langle0,29,31,47\rangle & \langle0,4,25,42\rangle & \langle0,1,11,15\rangle & \langle0,16,23,24\rangle & \\ \langle0,26,19,46\rangle \\
\end{array}$
\\
\hline $10$ & $\begin{array}{lllllllll}
\langle0,9,38,57\rangle & \langle0,22,30,33\rangle & \langle0,15,39,5\rangle & \langle0,20,12,21\rangle & \langle0,58,23,45\rangle & \langle0,3,2,17\rangle & \langle0,25,35,41\rangle &
\langle0,19,51,55\rangle &\\ \langle0,6,13,50\rangle & \langle0,28,46,26\rangle &
\end{array}$
\\
\hline $11$ & $\begin{array}{lllllllll}
\langle0,55,33,19\rangle & \langle0,37,65,49\rangle & \langle0,52,8,63\rangle & \langle0,67,60,2\rangle & \langle0,10,4,35\rangle & \langle0,59,66,48\rangle & \langle0,17,44,53\rangle &
\langle0,50,22,21\rangle &\\ \langle0,23,38,64\rangle & \langle0,26,5,56\rangle & \langle0,14,20,43\rangle
\end{array}$
\\
\hline $14$ & $\begin{array}{lllllllll}
\langle0,1,20,31\rangle & \langle0,7,47,68\rangle & \langle0,11,15,78\rangle & \langle0,14,71,80\rangle & \langle0,23,82,83\rangle & \langle0,5,29,41\rangle & \langle0,10,38,55\rangle &
\langle0,2,46,56\rangle &\\ \langle0,9,35,42\rangle & \langle0,12,18,70\rangle & \langle0,17,25,39\rangle & \langle0,34,37,50\rangle & \langle0,13,62,64\rangle & \langle0,21,48,53\rangle \\
\end{array}$
\\
\hline $17$ & $\begin{array}{lllllllll}
\langle0,1,81,86\rangle & \langle0,4,37,58\rangle & \langle0,8,65,74\rangle & \langle0,12,59,76\rangle & \langle0,22,71,72\rangle & \langle0,51,78,90\rangle & \langle0,3,19,29\rangle & \langle0,7,75,77\rangle & \\
\langle0,2,40,62\rangle & \langle0,5,11,46\rangle & \langle0,9,32,88\rangle & \langle0,17,30,61\rangle & \langle0,31,45,98\rangle & \langle0,69,89,93\rangle & \langle0,10,25,28\rangle & \langle0,21,55,63\rangle &  \\
\langle0,48,84,91\rangle \\
\end{array}$
\\
\hline
\end{tabular}
\end{table*}

\begin{lemma}
\label{2^{3t+1}} There exists a $[2,2]$-GDC$(6)$ of type
$2^{3t+1}$ with size $2t(3t+1)$ for each $t\geq 12$ and
$t\not\in \{14,17\}$.
\end{lemma}

\begin{proof} For each $t\geq 12$ and $t\not\in P$, take a
$[2,2]$-GDC$(6)$ of type $24^i30^j36^k42^l48^m$ with
$4i+5j+6k+7l+8m=t$ from Lemma~\ref{[2,2]-GDCa}. Adjoin two
ideal points, and fill in the groups together with the ideal
points with $[2,2]$-GDC$(6)$s of type $2^u$ for $u\in
\{13,16,19,22,25\}$ to obtain the desired GDC.

For each $t\in \{12,15,16,18,19,21,22,23,26,27,28,31,32,33\}$, take
the $[2,2]$-GDC$(6)$ constructed in Lemma~\ref{[2,2]-GDCb}. Adjoin
two ideal points, and fill in the groups to obtain the desired GDC.

For $t=13$, take a TD$(4,5)$ and apply the Fundamental Construction
with weight $4$ to obtain a $[2,2]$-GDC$(6)$ of type $20^4$. Then
fill in the groups with $[2,2]$-GDC$(6)$s of type $2^{10}$ to
obtain the required GDC.
\end{proof}

Summarizing the above results, we have:

\begin{theorem} \label{6t+2}
$A_3(8,6,[2,2])=5$, $A_3(14,6,[2,2])\geq 27$,
$A_3(6t+2,6,[2,2])=U(6t+2,6,[2,2])$ for each $t\geq 3$.
\end{theorem}

\subsection{Case of Length  $n\equiv 5\pmod{6}$}

\begin{lemma}
\label{s6t+5} $A_3(5,6,[2,2])=1$, $A_3(11,6,[2,2])=15$,
$A_3(17,6,[2,2])\geq40$.
\end{lemma}

\begin{proof} For the first equation, see Table~\ref{presult}. From \cite{ZhangGe:2010}, we have $A_3(11,6,[2,2])\leq
A_3(11,6,4)=15$. An optimal $(11,6,[2,2])_3$-code of size $15$ is also
constructed in \cite{ZhangGe:2010}.

A $(17,6,[2,2])_3$-code with $40$ codewords is constructed on
$\bbZ_{16}\cup\{\infty\}$, and is obtained by developing the following base
codewords $+4\pmod{16}$.
$$\begin{array}{llllll}
\langle1,5,8,14\rangle & \langle0,11,12,10\rangle & \langle0,6,7,4\rangle & \langle3,10,15,6\rangle &
\langle0,3,9,5\rangle & \langle1,6,\infty,12\rangle \\ \langle1,2,15,9\rangle & \langle0,\infty,13,15\rangle &
\langle2,4,5,6\rangle & \langle3,13,7,12\rangle\\
\end{array}$$
\end{proof}

\begin{lemma}
\label{s2^3t5^1a} There exists a $[2,2]$-GDC$(6)$ of type
$2^{3t}5^1$ with size $2t(3t+4)$ for each $3\leq t\leq 8$.
\end{lemma}

\begin{proof} Let $X_t=I_{6t+5}$, and ${\cal
G}_t=\{\{i,i+3t\}:0\leq i\leq
3t-1\}\cup\{\{6t,6t+1,6t+2,6t+3,6t+4\}\}$. Then $(X_t,{\cal
G}_t,{\cal C}_t)$ is a $[2,2]$-GDC$(6)$ of type $2^{3t}5^1$,
where ${\cal C}_3$ is obtained by developing the codewords in Table~\ref{t2^3t5^1} under the
automorphism group $G=\langle(0\ \ 3\ \ 6\ \ \cdots\ \ 15)(1\ \
4\ \ 7\ \ \cdots\ \ 16)$ $(2\ \ 5\ \ 8\ \ \cdots\ \ 17)(18\ \ 19\
\ 20)$ $(21\ \ 22)\rangle$, and for $t>3$, ${\cal C}_t$ is
obtained by developing the elements of $\bbZ_{6t}$ in the
codewords in Table~\ref{t2^3t5^1} under the automorphism group
$G=\langle(0\ \ 3\ \ 6\ \ \cdots\ \ 6t-3)(1\ \ 4\ \ 7\ \
\cdots\ \ 6t-2)(2\ \ 5\ \ 8\ \ \cdots\ \
6t-1)(6t)(6t+1)(6t+2)(6t+3)(6t+4)\rangle$.
\end{proof}

\begin{table*}
\scriptsize
\centering
\renewcommand{\arraystretch}{1}
\setlength\arraycolsep{3pt} \caption{Base Codewords of Small
$[2,2]$-GDC$(6)$s of Type $2^{3t}5^1$ in Lemma~\ref{s2^3t5^1a}}
\label{t2^3t5^1}
\begin{tabular}{c|l}
\hline $t$ & {\hfill Codewords \hfill}
\\
\hline $3$ & $\begin{array}{lllllllll}
\langle1,5,3,22\rangle & \langle0,21,7,11\rangle & \langle0,22,1,17\rangle & \langle1,20,2,15\rangle & \langle0,12,15,19\rangle & \langle0,20,4,14\rangle & \langle2,8,5,19\rangle &
\langle1,11,0,12\rangle &\\ \langle2,18,1,6\rangle & \langle2,4,12,22\rangle & \langle1,13,16,19\rangle & \langle0,5,10,16\rangle & \langle0,13,2,8\rangle &  \\
\end{array}$
\\
\hline $4$ & $\begin{array}{llllllll}
\langle2,25,18,22\rangle & \langle1,23,20,2\rangle & \langle0,22,25,5\rangle & \langle0,19,23,26\rangle & \langle2,0,10,28\rangle & \langle2,27,16,21\rangle & \langle0,14,13,24\rangle & \langle1,28,0,14\rangle \\
\langle2,26,15,7\rangle & \langle1,9,17,27\rangle & \langle2,8,19,17\rangle & \langle1,19,10,12\rangle & \langle0,4,1,7\rangle & \langle0,18,11,9\rangle & \langle1,24,11,15\rangle & \langle2,6,9,3\rangle \\
\end{array}$
\\
\hline $5$ & $\begin{array}{lllllllll}
\langle2,18,21,9\rangle & \langle0,19,2,33\rangle & \langle2,26,0,4\rangle & \langle1,6,31,11\rangle & \langle1,15,34,8\rangle & \langle2,32,19,24\rangle & \langle2,31,15,16\rangle &
\langle2,33,25,27\rangle &\\ \langle0,13,11,32\rangle & \langle1,27,24,4\rangle & \langle0,10,26,30\rangle & \langle0,1,9,22\rangle & \langle1,19,23,28\rangle & \langle2,34,22,3\rangle &
\langle0,12,20,6\rangle & \langle2,20,29,23\rangle &\\ \langle1,3,2,20\rangle & \langle2,30,28,12\rangle & \langle1,26,7,25\rangle &  &  \\
\end{array}$
\\
\hline $6$ & $\begin{array}{llllllllll}
\langle2,23,4,28\rangle & \langle0,28,23,37\rangle & \langle2,26,22,0\rangle & \langle2,37,24,13\rangle & \langle1,0,4,30\rangle & \langle0,22,36,20\rangle & \langle0,19,34,21\rangle & \langle1,24,20,38\rangle\\
\langle2,38,15,31\rangle & \langle0,31,27,2\rangle & \langle0,11,15,12\rangle & \langle1,21,39,26\rangle & \langle1,10,2,23\rangle & \langle2,36,33,25\rangle & \langle1,12,11,5\rangle &
\langle2,40,16,21\rangle &\\ \langle2,8,35,11\rangle & \langle2,39,9,10\rangle & \langle0,8,24,7\rangle & \langle1,25,31,22\rangle & \langle0,3,17,9\rangle & \langle0,10,26,40\rangle &  \\
\end{array}$
\\
\hline $7$ & $\begin{array}{lllllllll}
\langle2,15,38,0\rangle & \langle0,24,28,22\rangle & \langle2,29,35,19\rangle & \langle2,36,9,41\rangle & \langle2,33,11,26\rangle & \langle1,16,13,4\rangle & \langle1,10,21,32\rangle & \langle2,4,21,39\rangle \\
\langle2,18,31,43\rangle & \langle1,29,5,30\rangle & \langle2,30,42,22\rangle & \langle1,7,20,8\rangle & \langle2,12,45,1\rangle & \langle0,30,7,25\rangle & \langle1,46,38,24\rangle & \langle2,32,13,40\rangle \\
\langle1,45,0,26\rangle & \langle1,42,6,35\rangle & \langle1,43,11,3\rangle & \langle2,27,46,37\rangle & \langle0,44,16,2\rangle & \langle1,17,39,44\rangle & \langle1,25,9,15\rangle & \langle0,38,1,41\rangle \\
\langle0,36,33,3\rangle \\
\end{array}$
\\
\hline $8$ & $\begin{array}{lllllll}
\langle0,45,18,30\rangle & \langle2,50,42,25\rangle & \langle2,52,40,27\rangle & \langle1,3,45,7\rangle & \langle0,1,15,16\rangle & \langle1,4,34,22\rangle & \langle2,8,19,39\rangle \\
\langle0,19,47,41\rangle & \langle2,32,34,12\rangle & \langle0,43,32,27\rangle & \langle2,29,36,28\rangle & \langle0,22,48,17\rangle & \langle0,7,34,20\rangle & \langle2,5,38,17\rangle \\
\langle0,5,9,6\rangle & \langle1,9,51,32\rangle & \langle2,48,0,43\rangle & \langle0,36,44,14\rangle & \langle2,41,22,15\rangle & \langle1,21,10,12\rangle & \langle2,51,21,46\rangle \\
\langle16,39,26,50\rangle & \langle16,4,23,20\rangle & \langle16,2,10,7\rangle & \langle2,49,18,37\rangle & \langle0,13,38,49\rangle & \langle3,13,5,14\rangle & \langle0,31,29,52\rangle \\
\end{array}$
\\
\hline
\end{tabular}
\end{table*}

\begin{lemma}
\label{2^3t5^1} There exists a $[2,2]$-GDC$(6)$ of type
$2^{3t}5^1$ with size $2t(3t+4)$ for each $t\geq 12$ and
$t\not\in\{13,14,17\}$.
\end{lemma}

\begin{proof} For each $t\geq 12$ and $t\not\in P$, take a
$[2,2]$-GDC$(6)$ of type $24^i30^j36^k42^l48^m$ with
$4i+5j+6k+7l+8m=t$ from Lemma~\ref{[2,2]-GDCa}. Adjoin five
ideal points, and fill in the groups together with the ideal
points with $[2,2]$-GDC$(6)$s of type $2^{3s}5^1$ for
$s\in\{4,5,6,7,8\}$ to obtain the desired GDC.

For each $t\in \{12,15,16,18,19,21,22,23,26,27,28,31,$ $32,33\}$, take
the $[2,2]$-GDC$(6)$ constructed in Lemma~\ref{[2,2]-GDCb}. Adjoin
five ideal points, and fill in the groups together with the
ideal points with $[2,2]$-GDC$(6)$s of type $2^{3s}5^1$ for
$s\in\{3,4,5,6,7\}$ to obtain the desired GDC.
\end{proof}

\begin{lemma}
$A_3(6t+5,6,[2,2])\geq U(6t+5,6,[2,2])-1$ for each $t\geq 3$ and
$t\not\in\{9,10,11,14\}$; $A_3(6t+5,6,[2,2])\geq U(6t+5,6,[2,2])-2$ for $t=14$.
\end{lemma}

\begin{proof} For each $t\geq 3$ and
$t\not\in\{9,10,11,13,14,$ $17\}$, take a $[2,2]$-GDC$(6)$ of type $2^{3t}5^1$
from Lemmas~\ref{s2^3t5^1a} and \ref{2^3t5^1}, and fill in the group of
size $5$ with an optimal $(5,6,[2,2])_3$-code with one codeword
to get the desired code.

For each $t\in \{13,14,17\}$, take a $[2,2]$-GDC$(6)$ of type $g^u m^1$ with
$(g,u,m)\in\{(18,4,6),$ $(18,4,12), (24,4,6)\}$ from Lemma~\ref{[2,2]-GDC:sg^um^1}. Adjoin $5$ ideal points, fill
in the groups of size $g$ together with the ideal points with $[2,2]$-GDC$(6)$s of
type $2^{g/2} 5^1$, and fill in the group of size $m$ together with the ideal points
with a code of length $11$ or length $17$ from Lemma~\ref{s6t+5}.
\end{proof}

\begin{lemma}
\label{24t+5}
$A_3(24t+5,6,[2,2])=U(24t+5,6,[2,2])$ for each $t\geq 3$.
\end{lemma}

\begin{proof} Take a $[2,2]$-GDC$(6)$ of type $(6t+1)^4$ with size $2(6t+1)^2$ from Lemma~\ref{DM2GDCg^4}. Adjoin one extra point and fill in the groups with optimal codes of length $6t+2$ from Theorem~\ref{6t+2} to obtain the desired code.
\end{proof}

\begin{lemma}
$A_3(6t+5,6,[2,2])=U(6t+5,6,[2,2])$ for each $t \geq 130$ .
\end{lemma}

\begin{proof}  Take a TD$(7,m)$ with $m\geq 23$ and $m\not \in \{26,$ $30, 34,38,46,60\}$
from Theorem~\ref{TD} and apply the Fundamental Construction, assigning weight $6$ to
all points in the first five groups and weights $0$ or $6$ to the points in the last two groups.
Note that there exist $[2,2]$-GDC$(6)$s of type $6^s$ for $s\in \{5,6,7\}$ by Lemma~\ref{[2,2]-GDC:6^t}.
The result is a $[2,2]$-GDC$(6)$ of type $(6m)^5(6x)^1 72^1$ with $x\in[3,8]\cup [18,23]$. Adjoin five ideal points. Fill in the first six groups together with the ideal points with $[2,2]$-GDC$(6)$s of types $2^{3m}5^1$ and $2^{3x}5^1$, and fill in the group of size $72$ together with the five ideal points with an optimal $(77,6,[2,2])_3$-code from Lemma~\ref{24t+5}. The result is an optimal code of length $6t+5$ with $t=5m+x+12$, as desired.
\end{proof}

\begin{lemma}
$A_3(6t+5,6,[2,2])=U(6t+5,6,[2,2])$ for each $t \in \{54,55\}$ or $63 \leq  t \leq 129$.
\end{lemma}

\begin{proof}
For each $t\in \{54,55\}$, take a TD$(5,12)$ from Theorem~\ref{TD} and apply the Fundamental Construction, assigning weight $6$ to all points in the first four groups and weights $3$ or $6$ to  points in the last group. The result is a $[2,2]$-GDC$(6)$ of type $72^4 36^1$ or type $72^4 42^1$. Then adjoin five ideal points and fill in the groups with $[2,2]$-GDC$(6)$s of type $2^u 5^1$ for $u\in \{18,21,36\}$ and an optimal $(77,6,[2,2])_3$-code to obtain the desired code.

For each $63 \leq t \leq 76$, take a TD$(9,8)$ from Theorem~\ref{TD} and apply the Fundamental Construction, assigning weight $6$ to all points in the first six groups, weight $9$ to the points in the last group and weights $0$ or $6$ to the remaining points.
Note that there exist $[2,2]$-GDC$(6)$s of type $6^s 9^1$ for $s\in \{6,7,8\}$ by Lemma~\ref{[2,2]-GDC:s6^t9^1}. The result is a $[2,2]$-GDC$(6)$ of type $48^6(6x)^1 (6y)^1 72^1$ with $x,y \in \{0,3,4,5,6,7,8\}$. Adjoin five ideal points and fill in the groups to obtain an optimal code of length $6t+5$ with $t=60+x+y$.

For each $77\leq t \leq 100$ and $t\not \in \{85,86\}$, take a TD$(9,u)$ with $u\in \{9,11\}$ from Theorem~\ref{TD} and remove one point to redefine the groups to obtain a $\{9,u\}$-GDD of type $8^u (u-1)^1$. Then apply the Fundamental Construction, assigning weight $6$ to the points in the first $u-2$ groups of size $8$,  weights $0$ or $9$ to the points in the group of size $u-1$ and weights $0$ or $6$ to the remaining points. Note that there exist $[2,2]$-GDC$(6)$s of types $6^s 9^1$ with $s\in \{6,7,8\}$ and $6^s$ with $5\leq s\leq 11$ by Lemmas~\ref{[2,2]-GDC:6^t} and~\ref{[2,2]-GDC:s6^t9^1}. The result is a $[2,2]$-GDC$(6)$ of type $48^{u-2} (6x)^1 (6y)^1 72^1$ with $x,y \in \{0,3,4,5,6,7,8\}$. Adjoin five ideal points and fill in the groups to obtain an optimal code of length $6t+5$ with $t=8u+x+y-4$.

For each $t\in\{85,86\}$, take a TD$(8,11)$ from Theorem~\ref{TD} and remove one point to redefine the groups to obtain an $\{8,11\}$-GDD of type $7^{11} 10^1$. Then apply the Fundamental Construction, assigning weight $6$ to the points in the first ten groups of size $7$,  weights $0$ or $9$ to the points in the group of size $10$ and weights $0$ or $6$ to the remaining points. The result is a $[2,2]$-GDC$(6)$ of type $42^{10} (6x)^1 72^1$ with $x\in \{3,4\}$. Adjoin five ideal points and fill in the groups to obtain the desired optimal code.

Finally, for $101 \leq t \leq 129$, take a TD$(11,16)$ and apply the Fundamental Construction, assigning weight $6$ to the points in the first five groups and weights $0$ or $6$ to the remaining points. The result is a $[2,2]$-GDC$(6)$ of type $96^{5} (6x_1)^1 (6x_2)^1 \ldots (6x_5)^1 72^1$ with $x_1,x_2,\ldots,x_5 \in \{0,3,4,5,6,7,8\}$. Adjoin five ideal points and fill in the groups to  complete the proof.
\end{proof}

Summarizing the above results, we have:

\begin{theorem}
Let $Q^{(5)}=\{9,10,11\}$, $Q_{1}^{(5)}=\{3,4,5,6,7,$ $8,13,57,58,59,61,62\} \cup \{ t: 15\leq t\leq 53, t\not \equiv 0 \pmod{4}\}$, and $Q_{2}^{(5)}=\{2,14\}$. Then $A_3(5,6,[2,2])=1$, $A_3(11,6,[2,2])=15$, $A_3(6t+5,6,[2,2])= U(6t+5,6,[2,2])$ for each $t\geq 2$ and $t\not\in Q^{(5)} \cup Q_{1}^{(5)} \cup Q_2^{(5)}$. Furthermore, we have
\begin{enumerate}
\item $U(6t+5,6,[2,2])-1 \leq A_3(6t+5,6,[2,2])\leq U(6t+5,6,[2,2])$ for each $t\in Q_{1}^{(5)}$;
\item $U(6t+5,6,[2,2])-2 \leq A_3(6t+5,6,[2,2])\leq U(6t+5,6,[2,2])$ for each $t\in Q_{2}^{(5)}$.
\end{enumerate}
\end{theorem}

\subsection{Case of Length $n\equiv 4\pmod{6}$}

\begin{lemma}
\label{s6t+4} $A_3(6t+4,6,[2,2])=U(6t+4,6,[2,2])$ for $t=1$ or
$4\leq t\leq 9$.
\end{lemma}

\begin{proof} For $t=1$, see Table~\ref{presult}. For $4\leq t\leq 9$, let $X_t=\bbZ_{6t+4}$. Then $(X_t,{\cal
C}_t)$ is the desired optimal $(6t+4,6,[2,2])_3$-code, where
${\cal C}_t$ is obtained by developing the elements of
$\bbZ_{6t+4}$ in the codewords listed in Table~\ref{t6t+4}
$+2\pmod{6t+4}$.
\end{proof}

\begin{table*}
\scriptsize
\centering
\renewcommand{\arraystretch}{1}
\setlength\arraycolsep{3pt} \caption{Base Codewords of Small
Optimal $(6t+4,6,[2,2])_3$-Codes in Lemma~\ref{s6t+4}}
\label{t6t+4}
\begin{tabular}{c|l}
\hline $t$ & {\hfill Codewords \hfill}
\\
\hline $4$ & $\begin{array}{lllllllll}
\langle0,15,21,24\rangle & \langle0,6,9,11\rangle & \langle1,22,26,21\rangle & \langle1,3,11,15\rangle & \langle0,1,18,19\rangle & \langle1,6,23,16\rangle & \langle1,5,24,12\rangle & \langle1,4,17,2\rangle & \langle0,12,14,20\rangle \\
\end{array}$
\\
\hline $5$ & $\begin{array}{lllllllll}
\langle0,11,17,21\rangle & \langle1,19,10,18\rangle & \langle1,9,12,23\rangle & \langle1,33,30,20\rangle & \langle0,7,14,20\rangle & \langle0,24,23,15\rangle &
\langle0,8,1,3\rangle & \langle0,19,31,13\rangle & \\ \langle1,31,2,21\rangle & \langle0,28,12,16\rangle & \langle0,4,2,9\rangle &
\end{array}$
\\
\hline $6$ & $\begin{array}{lllllllll}
\langle0,19,10,12\rangle & \langle1,2,13,31\rangle & \langle1,26,0,3\rangle & \langle0,25,5,8\rangle & \langle0,4,32,13\rangle & \langle0,6,22,7\rangle & \langle1,19,11,17\rangle &
\langle0,3,27,31\rangle &\\ \langle1,5,18,12\rangle & \langle1,20,6,10\rangle & \langle0,2,20,35\rangle & \langle1,4,27,28\rangle & \langle1,7,15,36\rangle & \\
\end{array}$
\\
\hline $7$ & $\begin{array}{llllllll}
\langle1,35,40,24\rangle & \langle1,31,0,29\rangle & \langle1,9,10,12\rangle & \langle0,16,43,21\rangle & \langle0,22,17,7\rangle & \langle0,32,9,25\rangle & \langle1,33,14,20\rangle & \langle0,2,13,1\rangle \\
\langle1,18,7,21\rangle & \langle1,43,26,38\rangle & \langle1,23,3,41\rangle & \langle0,34,26,8\rangle & \langle0,18,37,33\rangle & \langle0,40,4,36\rangle & \langle1,11,8,32\rangle &  \\
\end{array}$
\\
\hline $8$ & $\begin{array}{lllllllll}
\langle0,50,3,11\rangle & \langle0,49,38,21\rangle & \langle1,33,37,43\rangle & \langle1,18,40,38\rangle & \langle1,9,30,2\rangle & \langle1,12,13,27\rangle & \langle1,24,32,49\rangle & \langle0,4,18,47\rangle \\
\langle0,17,12,39\rangle & \langle1,39,3,14\rangle & \langle1,20,51,19\rangle & \langle1,34,41,44\rangle & \langle1,7,35,16\rangle & \langle0,45,42,23\rangle & \langle0,24,16,9\rangle & \langle0,27,26,32\rangle & \\
\langle0,6,36,40\rangle \\
\end{array}$
\\
\hline $9$ & $\begin{array}{lllllllll}
\langle1,31,49,55\rangle & \langle0,1,17,39\rangle & \langle1,14,41,46\rangle & \langle1,18,36,11\rangle & \langle0,31,52,15\rangle & \langle1,6,35,4\rangle & \langle0,21,4,28\rangle &
\langle0,5,37,9\rangle &\\ \langle0,10,13,22\rangle & \langle1,26,51,52\rangle & \langle1,2,40,45\rangle & \langle1,10,24,16\rangle & \langle1,23,21,34\rangle & \langle0,16,23,35\rangle &
\langle1,48,20,30\rangle & \langle0,8,54,55\rangle &\\ \langle1,47,15,32\rangle & \langle1,7,56,9\rangle & \langle0,34,20,36\rangle &  &  \\
\end{array}$
\\
\hline
\end{tabular}
\end{table*}

\begin{lemma}
\label{6t+4a} $A(6t+4,6,[2,2])=U(6t+4,6,[2,2])$ for $t\geq 142$.
\end{lemma}

\begin{proof}
Take a TD$(8,m)$ with $m\geq 23$ and $m\not \in \{26,28,30,33,34,35,38,39,42,44,46,51,52,$ $54,58,60,62,66,68,74\}$ from Theorem~\ref{TD}. Apply the Fundamental Construction with weight $6$ to all points in the first $6$ groups, $x$ points in the seventh group, and weight $3$ to $3$ points in the last group. The other points are given weight $0$. Noting that there exist $[2,2]$-GDC$(6)$s of type $6^s$ for $s\in \{6,7\}$ by Lemma~\ref{[2,2]-GDC:6^t}, and $[2,2]$-GDC$(6)$s of type $6^s 3^1$ for $s\in \{6,7\}$ by Lemma~\ref{[2,2]-GDC:s6^t3^1}. The result is a $[2,2]$-GDC$(6)$ of type $(6m)^6(6x)^19^1$ for $x=0$ or $3\leq x\leq m$. Adjoin one ideal point. Fill in the first seven groups together with the ideal point with  optimal codes of lengths $6m+1$ and $6x+1$ from Theorem~\ref{6t+1}, and fill in the group of size $9$ together with the ideal point with an optimal $(10,6,[2,2])_3$-code. The result is an optimal code of length $6t+4$ with $t=6m+x+1$. That includes all integers $t\geq 142$.
\end{proof}

\begin{lemma}
\label{6t+4b} $A(6t+4,6,[2,2])=U(6t+4,6,[2,2])$ for $t=43$ or $46\leq t\leq 141$ and $t\neq 51$.
\end{lemma}

\begin{proof}
Take a TD$(k,m)$ from Theorem~\ref{TD} with $m\in \{7,8,9,13\}$ and $8\leq k\leq 12$, $k\leq m+1$. Apply the Fundamental Construction with weight $6$ to all the points in the first $6$ groups, weight $3$ to $3$ points in the last group, and weight $6$ to $x_i$ points in the remaining $k-7$ groups for $1\leq i \leq k-7$. The remaining points are given weight $0$. Take $x_i=0$ or $3\le x_i\leq m$. Noting that there exist $[2,2]$-GDC$(6)$s of type $6^s$ for $6\leq s\leq 11$ by Lemma~\ref{[2,2]-GDC:6^t}, and $[2,2]$-GDC$(6)$s of type $6^s 3^1$ for $6\leq s\leq 11$ by Lemma~\ref{[2,2]-GDC:s6^t3^1}. The result is a $[2,2]$-GDC$(6)$ of type
$(6m)^6 (6x_1)^1 \ldots (6x_{k-7})^1 9^1.$
Adjoin one ideal point. Fill in the first $k-1$ groups together with the ideal point with optimal codes of lengths $6m+1$ and $6x_i+1$ from Theorem~\ref{6t+1}, and fill in the group of size $9$ together with the ideal point with an optimal $(10,6,[2,2])_3$-code. The result is an optimal code of length $6t+4$ with $t = 6m+\sum_{i=1}^{k-7}x_i+1$.

Taking  $k=8$ and $m=7$, we obtain $t\in \{43\}\cup [46,50]$. Taking  $k=9$ and $m=8$, we obtain $t\in [52,65]$. Taking $k=10$ and $m=9$, we obtain an optimal code of length $t\in [66,82]$. Taking $k=12$ and $m=13$, we obtain $t\in [83,141]$.
\end{proof}

\begin{lemma}
\label{6t+4c} $A(6t+4,6,[2,2])=U(6t+4,6,[2,2])$ for $t\in \{24\}\cup
[32,34]\cup [36,42] \cup \{44\}$.
\end{lemma}

\begin{proof}
Take a TD$(5,m)$ with $m\in \{5,7,8,9\}$ from Theorem~\ref{TD}. Apply the Fundamental Construction with weight $6$ to all the points in the first $4$ groups, $x$ points in the last group, and weight $3$ to $y$ points in the last group, such that $x+y=m$. Noting that there exists a $[2,2]$-GDC$(6)$ of type $6^5$ by Lemma~\ref{[2,2]-GDC:6^t}, and a $[2,2]$-GDC$(6)$ of type $6^4 3^1$ by Lemma~\ref{[2,2]-GDC:s6^t3^1}. The result is a $[2,2]$-GDC$(6)$ of type $(6m)^4(6x+3y)^1$. Adjoin one ideal point. Fill in the first $4$ groups together with the ideal point with optimal codes of length $6m+1$ from Theorem~\ref{6t+1}, and fill in the group of size $6x+3y$ together with the ideal point with an optimal $(6x+3y+1,6,[2,2])_3$-code from Lemma~\ref{s6t+4}. The result is an optimal code of length $6t+4$ with $t=4m+x+\frac{y-1}{2}$.
For each desired $t$, the parameters $(m,x,y)$ and the code of length $s=6x+3y+1$ to be filled in are given in Table~\ref{t6t+4c}.
\end{proof}

\begin{table}
\caption{The Parameters for Lemma~\ref{6t+4c}} \label{t6t+4c}
\centering
\begin{tabular}{c|c|c||c|c|c||c|c|c}
\hline
$t$ & $(m,x,y)$ & $s$ & $t$ & $(m,x,y)$ & $s$ & $t$ & $(m,x,y)$ & $s$ \\
\hline
$24$ & $(5,4,1)$ & $28$ & $32$ & $(7,2,5)$ & $28$ & $33$ & $(7,4,3)$ & $34$ \\
\hline
$34$ & $(7,6,1)$ & $40$ & $36$ & $(8,1,7)$ & $28$ & $37$ & $(8,3,5)$ & $34$ \\
\hline
$38$ & $(8,5,3)$ & $40$ & $39$ & $(8,7,1)$ & $46$ & $40$ & $(9,0,9)$ & $28$ \\
\hline
$41$ & $(9,2,7)$ & $34$ & $42$ & $(9,4,5)$ & $40$ & $44$ & $(9,8,1)$ & $52$\\
\hline
\end{tabular}
\end{table}

\begin{lemma}
\label{6t+4d} $A(6t+4,6,[2,2])=U(6t+4,6,[2,2])$ for each $t\in \{10,
11, 13, 16, 17, 18, 19, 21,$ $22, 23, 25, 26, 28, 31, 45, 51\}.$
\end{lemma}

\begin{proof}
For $t=11$, take a $[2,2]$-GDC$(6)$ of type $10^7$ from Lemma~\ref{[2,2]-GDC:sg^u}. Fill in the groups with optimal $(10,6,[2,2])_3$-codes to obtain the desired code.

For $t=17$, take a $[2,2]$-GDC$(6)$ of type $24^4 9^1$ from
Lemma~\ref{[2,2]-GDC:sg^um^1}. Adjoin an ideal point, and fill in
each group together with the ideal point with an optimal
$(25,6,[2,2])_3$-code or an optimal $(10,6,[2,2])_3$-code to obtain
the desired code. For $t=23$, we proceed similarly, starting instead
with a  $[2,2]$-GDC$(6)$ of type $18^6 33^1$ from
Lemma~\ref{[2,2]-GDC:sg^um^1}.

For each desired $t\in \{10,13,16,18,19,21,22,25,26,28,$ $31,45,51\}$, we list the parameters to obtain the optimal code of length $6t+4$ in Table~\ref{t6t+4d}. We inflate a $[2,2]$-GDC$(6)$ of type $g^um^1$ (from ``Source'') with weight $w$, adjoin $a$ ideal point with $a=0$ or $1$, and fill in the groups together with the ideal point with optimal codes of length $s\in S$ to obtain the desired code.
\end{proof}

\begin{table}
\caption{The Parameters for Lemma~\ref{6t+4d}} \label{t6t+4d}
\centering
\begin{tabular}{c|c|c|c|c|c}
\hline
$t$ & $n$ & $g^um^1\times w$ & Source & $a$ & $S$ \\
\hline
$10$ & $64$ & $3^7\times 3$ & Lemma~\ref{[2,2]-GDC:sg^u} & $1$ & $10$ \\
$13$ & $82$ & $6^43^1\times 3$ & Lemma~\ref{[2,2]-GDC:s6^t3^1} & $1$ & $19,10$ \\
$16$ & $100$ & $3^{11}\times 3$ & Lemma~\ref{[2,2]-GDC:sg^u} & $1$ & $10$ \\
$18$ & $112$ & $4^4\times 7$ & Lemma~\ref{DM2GDCg^4} & $0$ & $28$ \\
$19$ & $118$ & $3^{13}\times 3$ & Lemma~\ref{[2,2]-GDC:sg^u} & $1$ & $10$ \\
$21$ & $130$ & $2^{13}\times 5$ & Lemma~\ref{s2^{3t+1}} & $0$ & $10$ \\
$22$ & $136$ & $6^73^1\times 3$ & Lemma~\ref{[2,2]-GDC:s6^t3^1} & $1$ & $19,10$ \\
$25$ & $154$ & $6^79^1\times 3$ & Lemma~\ref{[2,2]-GDC:s6^t9^1} & $1$ & $19,28$ \\
$26$ & $160$ & $4^4\times 10$ & Lemma~\ref{DM2GDCg^4} & $0$ & $40$ \\
$28$ & $172$ & $6^89^1\times 3$ & Lemma~\ref{[2,2]-GDC:s6^t9^1} & $1$ & $19,28$ \\
$31$ & $190$ & $3^7\times 9$ & Lemma~\ref{[2,2]-GDC:sg^u} & $1$ & $28$ \\
$45$ & $274$ & $3^7\times 13$ & Lemma~\ref{[2,2]-GDC:sg^u} & $1$ & $40$ \\
$51$ & $310$ & $1^{31}\times 10$ & Lemma~\ref{s6t+1} & $0$ & $10$ \\
\hline
\end{tabular}
\end{table}

\begin{lemma}
\label{6t+4e} $A(6t+4,6,[2,2])=U(6t+4,6,[2,2])$ for $t\in \{29,35\}$.
\end{lemma}

\begin{proof}
For $t=29$, take a $[2,2]$-GDC$(6)$ of type $6^7$ from Lemma~\ref{[2,2]-GDC:6^t}; apply the Inflation Construction with weight $4$, using $4$-MGDDs of type $4^4$ (see \cite{DS:1983}) as input designs, to obtain a $4$-DGDD of type $(24,6^4)^7$ with the CCC property. Adjoin $9$ ideal points, and fill in $[2,2]$-GDC$(6)$s of type $6^7 9^1$ to obtain a $[2,2]$-GDC$(6)$ of type $24^7 9^1$. Adjoin one more ideal point, and fill in the groups together with the ideal point with optimal codes of lengths $25$ and $10$ to obtain the desired code.

For $t=35$, take a $[2,2]$-GDC$(6)$ of type $6^7$ and remove all the points in the last group; apply the Inflation Construction with weight $5$, using $4$-MGDDs of type $5^4$ (see \cite{DS:1983}) and resolvable $3$-MGDDs of type $5^3$ (see \cite{WTD:2007}) as input designs, to obtain a $\{3,4\}$-DGDD of type $(30,6^5)^6$ with the CCC property, whose triples fall into $48$ parallel classes. Adjoin $24$ infinite points to complete the parallel classes, and then adjoin further $9$ ideal points, fill in a $[2,2]$-GDC$(6)$ of type $6^69^1$ to obtain a $[2,2]$-GDC$(6)$ of type $30^633^1$. Adjoin one more point, and fill in the groups together with the ideal point with optimal codes of lengths $31$ and $34$ to obtain the desired code.
\end{proof}

\begin{lemma}
\label{[2,2]-GDC-1^27 4^1}
There exists a $[2,2]$-GDC$(6)$ of type $1^{27} 4^1$ with size $153$.
\end{lemma}

\begin{proof}
Let $X=(\bbZ_{9} \times \{0, 1, 2\})\cup (\bbZ_{3} \times \{3\}) \cup \{\infty\}$. The point set $(\bbZ_{3} \times \{3\}) \cup \{\infty\}$ forms the group of size 4. Define $\alpha: X \rightarrow X$ as $x_y \rightarrow (x+1)_y$ where the addition is modulo 9 if $y \in \{0, 1, 2\}$, and modulo 3 if $y \in \{3\}$. The point $\infty$ is fixed by $\alpha$.
Develop the following 17 base codewords with $\alpha$:
$$\begin{array}{llllll}
\langle7_2, 0_3, 3_2, 5_0\rangle &
\langle5_1, 0_3, 1_0, 3_0\rangle &
\langle8_2, 0_3, 6_1, 4_1\rangle &
\langle5_0, 0_2, 0_3, 8_2\rangle &
\langle1_1, 0_1, 0_3, 5_1\rangle &
\langle3_0, 7_0, 0_3, 4_2\rangle \\
\langle5_1, \infty, 0_0, 4_2\rangle &
\langle0_0, 7_2, \infty, 0_1\rangle &
\langle5_1, 7_1, 7_0, 6_0\rangle &
\langle8_1, 0_2, 1_2, 4_2\rangle &
\langle7_1, 4_1, 2_2, 4_2\rangle &
\langle1_0, 7_0, 6_2, 3_1\rangle \\
\langle0_2, 6_1, 3_0, 0_0\rangle &
\langle1_2, 8_2, 5_0, 0_0\rangle &
\langle5_2, 8_2, 0_1, 8_1\rangle &
\langle4_0, 3_0, 1_1, 7_1\rangle &
\langle0_0, 2_0, 1_1, 2_2\rangle &
\end{array}$$
\end{proof}

\begin{lemma}
\label{6t+4e} $A(6t+4,6,[2,2])=U(6t+4,6,[2,2])$ for $t\in \{27,30\}$.
\end{lemma}

\begin{proof}
For each $t\in \{27,30\}$, take a $[2,2]$-GDC$(6)$ of type $9^{6}$ or type $9^{5} 15^1$ from Lemma~\ref{[2,2]-GDC:sg^um^1}. Apply the Inflation Construction with weight $3$ to obtain a $[2,2]$-GDC$(6)$ of type $27^{6}$ or type $27^{5} 45^1$. Then adjoin four ideal points and fill in the groups together with these ideal points with $[2,2]$-GDC$(6)$s of type $1^{27} 4^1$ and an optimal code of length $31$ or length $49$ to obtain the desired code.
\end{proof}

Combing the above lemmas, we have the following result:

\begin{theorem}
\label{6t+4}  $A_3(6t+4,6,[2,2])=U(6t+4,6,[2,2])$ for each $t\geq
1$, $t\not\in\{2,3,12,14,15,20\}$.
\end{theorem}

\subsection{Case of Length  $n\equiv 3\pmod{6}$}

\begin{theorem}
\label{s6t+3} $A_3(6t+3,6,[2,2])=U(6t+3,3)$ for each $t\geq 1$.
\end{theorem}

\begin{proof} For each $t\geq 1$, $t\not\in\{2,3,12,14,15,20\}$, remove one point and the related codewords from an optimal $(6t+4,6,[2,2])_3$-code from Theorem~\ref{6t+4} to obtain the desired code.

For $t\in \{2,14,20\}$, let $X_t=\bbZ_{6t+3}$. Then $(X_t,{\cal C}_t)$ is the desired optimal $(6t+3,6,[2,2])_3$-code, where ${\cal C}_t$ is
obtained by developing the elements of $\bbZ_{6t+3}$ in the following codewords $+1\pmod{6t+3}$.

\noindent $t=2$: $\langle0,3,4,14\rangle$ $\langle0,5,7,13\rangle$

\noindent $t=14$:
$$\begin{array}{llllll}
\langle0,24,47,61\rangle & \langle0,70,74,9\rangle & \langle0,41,16,36\rangle & \langle0,22,25,53\rangle  &
\langle0,28,83,66\rangle & \langle0,39,52,50\rangle \\ \langle0,57,51,75\rangle & \langle0,14,49,5\rangle  &
\langle0,58,68,27\rangle & \langle0,72,6,45\rangle & \langle0,20,54,84\rangle & \langle0,43,33,32\rangle \\
\langle0,85,69,40\rangle& \langle0,1,80,8\rangle \\
\end{array}$$

\noindent $t=20$:
$$\begin{array}{llllll}
\langle0,16,46,63\rangle & \langle0,34,122,9\rangle & \langle0,10,43,74\rangle & \langle0,92,85,7\rangle &
\langle0,27,71,87\rangle & \langle0,17,2,8\rangle \\ \langle0,5,109,82\rangle & \langle0,23,37,55\rangle &
\langle0,24,22,94\rangle & \langle0,6,21,79\rangle & \langle0,72,35,69\rangle & \langle0,58,1,84\rangle  \\
\langle0,18,54,59\rangle & \langle0,11,102,3\rangle & \langle0,45,93,97\rangle & \langle0,4,57,80\rangle  &
\langle0,13,42,62\rangle & \langle0,56,81,68\rangle \\ \langle0,20,95,39\rangle & \langle0,83,50,61\rangle  &
\end{array}$$

For $t=3$, the $[2,2]$-GDC$(6)$ of type $3^7$ constructed in Lemma~\ref{[2,2]-GDC:sg^u} is the desired code.

For $t\in \{12,15\}$, take a $[2,2]$-GDC$(6)$ of type $18^4$ (see
Lemma~\ref{[2,2]-GDC:sg^u}) or type $18^5$. Adjoin $3$ ideal points,
and fill in the groups together with these ideal points with
$[2,2]$-GDC$(6)$s of type $3^7$ to obtain the desired code.
\end{proof}

\section{Determining the Value of $A_3(n,6,[3,1])$}

In this section,
 we focus on the determination for the exact values of  $A_3(n,6,[3,1])$ for all positive integers $n$.

\subsection{Some $[3,1]$-GDC(6)s}

\begin{lemma}
\label{s3^3t+1} There is a $[3,1]$-GDC$(6)$ of type $3^{3t+1}$
with size $3t(3t+1)$ for each $t\in \{2,6\}$.
\end{lemma}

\begin{proof}
Let $X_t=\bbZ_{9t+3}$, and ${\cal G}_t=\{\{0,3t+1,6t+2\}+i:0\leq
i\leq 3t\}$. Then $(X_t,{\cal G}_t,{\cal C}_t)$ is a
$[3,1]$-GDC$(6)$ of type $3^{3t+1}$ with size $3t(3t+1)$, where
 ${\cal C}_2$ is obtained by developing the elements of
$\bbZ_{21}$ in the following codewords $+3\pmod{21}$, and ${\cal
C}_6$ is obtained by developing the elements of $\bbZ_{57}$ in
the following codewords $+1\pmod{57}$.

\noindent $t=2$: $\langle0,1,13,18\rangle$
$\langle1,4,14,5\rangle$ $\langle0,10,12,16\rangle$
$\langle2,5,7,1\rangle$ $\langle0,2,8,3\rangle$
$\langle0,6,17,5\rangle$

\noindent $t=6$: $\langle0,9,7,36\rangle$
$\langle0,1,6,21\rangle$ $\langle0,26,8,30\rangle$
$\langle0,3,44,40\rangle$ $\langle0,11,43,28\rangle$
$\langle0,10,33,45\rangle$
\end{proof}

\begin{lemma}
\label{s6^t} There exists a $[3,1]$-GDC$(6)$ of type $6^{3t+1}$ with size $12t(3t+1)$
for each $2\leq t\leq 5$.
\end{lemma}

\begin{proof}
Let $X_t=\bbZ_{18t+6}$, and ${\cal
G}_t=\{\{0,3t+1,6t+2,\dots,15t+5\}+i:0\leq i\leq 3t\}$. Then
$(X_t,{\cal G}_t,{\cal C}_t)$ is a $[3,1]$-GDC$(6)$ of type
$6^{3t+1}$, where ${\cal C}_t$ is obtained by developing the
elements of $\bbZ_{18t+6}$ in the codewords in Table~\ref{t6^t}
$+1\pmod{18t+6}$.
\end{proof}

\begin{table*}
\centering
\renewcommand{\arraystretch}{1}
\setlength\arraycolsep{3pt} \caption{Base Codewords of Small
$[3,1]$-GDC$(6)$s of type $6^{3t+1}$ in Lemma~\ref{s6^t}}
\label{t6^t}
\begin{tabular}{c|l}
\hline $t$ & {\hfill Codewords \hfill} \\
\hline $2$ & $\begin{array}{lllllll}
\langle0,38,32,1\rangle & \langle0,12,34,29\rangle & \langle0,23,26,25\rangle & \langle0,9,27,40\rangle &
\end{array}$ \\
\hline $3$ & $\begin{array}{lllllll}
\langle0,47,43,15\rangle & \langle0,21,23,35\rangle & \langle0,22,55,41\rangle & \langle0,34,31,25\rangle & \langle0,36,44,45\rangle & \langle0,7,18,6\rangle &
\end{array}$ \\
\hline $4$ & $\begin{array}{llllll}
\langle0,19,47,33\rangle & \langle0,49,69,7\rangle & \langle0,34,40,10\rangle & \langle0,46,11,23\rangle & \langle0,1,74,62\rangle & \langle0,2,53,70\rangle \\ \langle0,37,22,30\rangle & \langle0,3,21,45\rangle \\
\end{array}$ \\
\hline $5$ & $\begin{array}{llllll}
\langle0,54,1,87\rangle & \langle0,71,12,63\rangle & \langle0,7,34,65\rangle & \langle0,76,24,29\rangle & \langle0,92,19,41\rangle & \langle0,11,14,81\rangle \\ \langle0,57,36,83\rangle &
\langle0,66,94,8\rangle & \langle0,35,18,13\rangle & \langle0,6,56,15\rangle &
\end{array}$ \\
\hline
\end{tabular}
\end{table*}

\begin{lemma}
\label{s9^t} There is a $[3,1]$-GDC$(6)$ of type $9^t$ with size
$9t(t-1)$ for each $t\in \{4,5,6,8,9,11,$ $12,14,15,18,23\}$.
\end{lemma}

\begin{proof}
Let $X_t=\bbZ_{9t}$, and ${\cal
G}_t=\{\{0,t,2t,\dots,8t\}+i:0\leq i\leq t-1\}$. Then
$(X_t,{\cal G}_t,{\cal C}_t)$ is a $[3,1]$-GDC$(6)$ of type
$9^t$ with size $9t(t-1)$, where ${\cal C}_4$ is obtained by
developing the elements of $\bbZ_{36}$ in the codewords in
Table~\ref{t9^t} $+6\pmod{36}$, ${\cal C}_5$ is obtained by
developing the elements of $\bbZ_{45}$ in the codewords in
Table~\ref{t9^t} $+3\pmod{45}$, and  ${\cal C}_t$ with $t\geq 6$ is
obtained by developing the elements of $\bbZ_{9t}$ in the
codewords in Table~\ref{t9^t} $+1\pmod{9t}$.
\end{proof}

\begin{table*}
\centering
\renewcommand{\arraystretch}{1}
\setlength\arraycolsep{3pt} \caption{Base Codewords of Small
$[3,1]$-GDC$(6)$s of Type $9^t$ in Lemma~\ref{s9^t}}
\label{t9^t}
\begin{tabular}{c|l}
\hline $t$ & {\hfill Base Codewords \hfill}
\\
\hline $4$ & $\begin{array}{llllll}
\langle4,2,35,29\rangle & \langle5,19,34,16\rangle & \langle3,26,0,21\rangle & \langle0,15,5,18\rangle & \langle0,25,35,34\rangle & \langle4,15,17,6\rangle \\ \langle5,2,24,11\rangle &
\langle3,22,16,25\rangle & \langle2,23,25,32\rangle & \langle5,27,26,32\rangle & \langle4,23,30,5\rangle & \langle1,24,18,15\rangle \\ \langle2,13,16,31\rangle & \langle4,3,9,18\rangle &
\langle2,28,21,19\rangle & \langle2,7,0,9\rangle & \langle0,1,31,22\rangle & \langle1,2,27,20\rangle
\end{array}$
\\
\hline $5$ & $\begin{array}{llllll}
\langle1,17,35,9\rangle & \langle0,16,29,37\rangle & \langle2,8,44,1\rangle & \langle1,34,37,38\rangle & \langle0,23,36,19\rangle & \langle1,32,39,8\rangle \\ \langle0,4,31,3\rangle & \langle0,11,44,42\rangle &
\langle1,7,29,33\rangle & \langle1,3,20,44\rangle & \langle0,6,27,8\rangle & \langle0,12,34,13\rangle
\end{array}$
 \\
\hline $6$ & $\begin{array}{llllll}
\langle0,1,35,38\rangle & \langle0,14,29,46\rangle & \langle0,11,52,8\rangle & \langle0,26,47,16\rangle & \langle0,4,9,31\rangle \\
\end{array}$
 \\
\hline $8$ & $\begin{array}{llllll}
\langle0,58,71,19\rangle & \langle0,55,66,36\rangle & \langle0,26,28,67\rangle & \langle0,43,65,31\rangle & \langle0,37,47,52\rangle & \langle0,3,21,12\rangle \\ \langle0,4,49,34\rangle
\end{array}$
 \\
\hline $9$ & $\begin{array}{llllll}
\langle0,14,55,61\rangle & \langle0,22,71,20\rangle & \langle0,8,58,60\rangle & \langle0,11,48,77\rangle & \langle0,53,56,68\rangle & \langle0,35,74,69\rangle \\ \langle0,64,80,4\rangle & \langle0,19,43,13\rangle
\end{array}$
 \\
\hline $11$ & $\begin{array}{llllll}
\langle0,40,56,93\rangle & \langle0,41,51,98\rangle & \langle0,25,90,8\rangle & \langle0,2,29,71\rangle & \langle0,45,84,91\rangle & \langle0,63,95,26\rangle\\ \langle0,13,31,92\rangle & \langle0,23,35,73\rangle &
\langle0,14,19,20\rangle & \langle0,3,24,52\rangle
\end{array}$
 \\
\hline $12$ & $\begin{array}{llllll}
\langle0,80,83,46\rangle & \langle0,73,93,107\rangle & \langle0,44,89,22\rangle & \langle0,91,95,52\rangle & \langle0,10,100,59\rangle & \langle0,9,79,85\rangle \\ \langle0,75,77,23\rangle &
\langle0,82,103,32\rangle & \langle0,42,53,92\rangle & \langle0,27,57,43\rangle & \langle0,40,101,102\rangle
\end{array}$
\\
\hline $14$ & $\begin{array}{llllll}
\langle0,8,47,77\rangle & \langle0,45,46,23\rangle & \langle0,10,29,63\rangle & \langle0,91,111,73\rangle & \langle0,89,110,41\rangle & \langle0,3,62,71\rangle \\ \langle0,52,83,22\rangle &
\langle0,27,93,76\rangle & \langle0,75,101,92\rangle & \langle0,119,124,48\rangle & \langle0,86,90,18\rangle & \langle0,12,94,6\rangle \\ \langle0,13,24,85\rangle
\end{array}$
\\
\hline $15$ & $\begin{array}{lllll}
\langle0,55,116,88\rangle & \langle0,66,76,104\rangle & \langle0,23,123,44\rangle & \langle0,130,134,82\rangle & \langle0,77,85,51\rangle \\ \langle0,16,78,84\rangle & \langle0,108,111,67\rangle &
\langle0,22,42,114\rangle & \langle0,7,36,70\rangle & \langle0,49,89,102\rangle \\ \langle0,124,133,41\rangle & \langle0,14,32,79\rangle & \langle0,25,96,122\rangle & \langle0,37,118,31\rangle
\end{array}$
\\
\hline $18$ & $\begin{array}{lllll}
\langle0,28,91,152\rangle & \langle0,45,49,52\rangle & \langle0,70,79,123\rangle & \langle0,80,140,115\rangle & \langle0,46,65,39\rangle \\ \langle0,37,112,159\rangle & \langle0,156,157,20\rangle &
\langle0,73,132,38\rangle & \langle0,14,149,81\rangle & \langle0,42,111,142\rangle \\ \langle0,104,138,40\rangle & \langle0,11,66,95\rangle & \langle0,32,48,110\rangle & \langle0,119,160,129\rangle &
\langle0,56,141,12\rangle \\ \langle0,139,147,86\rangle & \langle0,17,74,150\rangle
\end{array}$
\\
\hline $23$ & $\begin{array}{lllll}
\langle0,72,189,190\rangle & \langle0,14,116,17\rangle & \langle0,76,144,175\rangle & \langle0,19,169,176\rangle & \langle0,24,59,121\rangle \\ \langle0,103,147,168\rangle & \langle0,48,101,98\rangle &
\langle0,28,141,70\rangle & \langle0,85,96,185\rangle & \langle0,120,153,32\rangle \\ \langle0,64,198,206\rangle & \langle0,52,177,58\rangle & \langle0,15,56,201\rangle & \langle0,43,123,83\rangle &
\langle0,4,55,162\rangle \\ \langle0,61,191,110\rangle & \langle0,5,133,170\rangle & \langle0,34,129,109\rangle & \langle0,10,36,81\rangle & \langle0,178,180,200\rangle \\ \langle0,47,114,39\rangle &
\langle0,182,194,124\rangle
\end{array}$
\\
\hline
\end{tabular}
\end{table*}

\begin{theorem}
\label{9^t} There is a $[3,1]$-GDC$(6)$ of type $9^t$ with size
$9t(t-1)$ for each
$t\geq 4$.
\end{theorem}

\begin{proof}
For $t\in \{4,5,6,8,9,11,12,14,15,18,23\}$, the desired codes
are constructed in Lemma~\ref{s9^t}.
For each $t\in \{7,19\}$, the desired code can be obtained by
inflating a $[3,1]$-GDC$(6)$ of type $3^t$ from
Lemma~\ref{s3^3t+1} with weight $3$.
For $t=10$, inflate a $[3,1]$-GDC$(6)$ of type $1^{10}$ with weight
$9$ to obtain the desired code.
For each $t\geq 10$ and $t\not\in \{10,11,12,14,15,18,19,$
$23\}$, take a $(t,\{4,5,6,7,8,9\},1)$-PBD from
Theorem~\ref{PBD4-9}, and apply the Fundamental Construction with weight
$9$ to obtain a $[3,1]$-GDC$(6)$ of type $9^t$.
\end{proof}

\begin{lemma}
\label{s15^t} There is a $[3,1]$-GDC$(6)$ of type $15^4$ with size $300$.
\end{lemma}

\begin{proof}
Let $X=\bbZ_{60}$, and ${\cal
G}=\{\{i,i+4,i+8,\dots,i+56\}:0\leq i\leq 3\}$. Then $(X,{\cal
G},{\cal C})$ is a $[3,1]$-GDC$(6)$ of type $15^4$, where
${\cal C}$ is obtained by developing the elements of $\bbZ_{60}$
in the following codewords $+2\pmod{60}$.
$$\begin{array}{llllll}
\langle1,15,40,42\rangle & \langle0,39,13,58\rangle & \langle0,3,34,25\rangle &
\langle0,54,37,27\rangle & \langle1,0,7,10\rangle & \langle1,12,50,31\rangle \\
\langle0,14,45,23\rangle & \langle1,56,14,11\rangle & \langle1,2,19,52\rangle &
\langle0,55,53,30\rangle &
\end{array}$$
\end{proof}

\begin{lemma}
\label{GDC-18^4} There is a $[3,1]$-GDC$(6)$ of type $18^4$ with size $432$.
\end{lemma}

\begin{proof}
Let $X=\bbZ_{72}$, and ${\cal
G}=\{\{i,i+4,i+8,\dots,i+68\}:0\leq i\leq 3\}$. Then $(X,{\cal
G},{\cal C})$ is a $[3,1]$-GDC$(6)$ of type $18^4$, where
${\cal C}$ is obtained by developing the elements of $\bbZ_{72}$
in the following codewords $+1\pmod{72}$.
$$\begin{array}{llllll}
\langle43, 69, 22, 20\rangle &
\langle62, 0, 53, 59\rangle &
\langle18, 1, 55, 68\rangle &
\langle31, 65, 4, 70\rangle &
\langle66, 24, 25, 27\rangle &
\langle0, 14, 57, 7\rangle \\
\end{array}$$
\end{proof}

\begin{lemma}
\label{GDC-39^4} There is a $[3,1]$-GDC$(6)$ of type $39^4$ with size $2028$.
\end{lemma}

\begin{proof}
Let $X=\bbZ_{156}$, and ${\cal G}=\{\{i,i+4,i+8,\dots,i+152\}:0\leq i\leq 3\}$. Then $(X,{\cal
G},{\cal C})$ is a $[3,1]$-GDC$(6)$ of type $39^4$, where
${\cal C}$ is obtained by developing the elements of $\bbZ_{156}$
in the following codewords $+1\pmod{156}$.
$$\begin{array}{lllll}
\langle95, 37, 0, 30\rangle &
\langle112, 97, 151, 22\rangle &
\langle71, 33, 122, 44\rangle &
\langle8, 131, 81, 82\rangle &
\langle56, 13, 98, 103\rangle \\
\langle44, 57, 66, 43\rangle &
\langle58, 99, 64, 113\rangle &
\langle128, 49, 18, 47\rangle &
\langle46, 28, 25, 35\rangle &
\langle17, 76, 110, 19\rangle \\
\langle23, 4, 49, 150\rangle &
\langle79, 17, 148, 74\rangle &
\langle0, 17, 103, 126\rangle &
\end{array}$$
\end{proof}

\begin{lemma}
\label{GDC-9^10(18&27)^1}
There exists a $[3,1]$-GDC$(6)$ of type $9^{10} m^1$ with size $10(81+2m)$ for each $m \in \{18,27\}$.
\end{lemma}

\begin{proof}
For the type $9^{10} 18^1$, let $X=I_{108}$, and ${\cal
G}=\{\{i,i+10,i+20,\dots,i+80\}:0\leq i <10\}\cup\{\{90,91,92,\dots,107\}\}$. Then $(X,{\cal
G},{\cal C})$ is a $[3,1]$-GDC$(6)$ of type $9^{10} 18^1$, where
${\cal C}$ is obtained by developing the following codewords under the
automorphism group $G=\langle(0\ \ 1\ \ 2\ \ \cdots\ \ 89)(90\ \ 91\ \ \cdots\ \ 98)$ $(99\ \
100\ \ \cdots\ \ 107)\rangle$.
$$\begin{array}{llllll}
\langle90, 31, 33, 0\rangle &
\langle90, 48, 26, 37\rangle &
\langle90, 43, 86, 74\rangle &
\langle9, 47, 82, 90\rangle &
\langle99, 10, 16, 74\rangle &
\langle99, 54, 8, 3\rangle \\
\langle99, 6, 40, 68\rangle &
\langle39, 23, 20, 99\rangle &
\langle86, 14, 21, 47\rangle &
\langle13, 12, 54, 66\rangle &
\langle58, 37, 45, 82\rangle &
\langle0, 4, 27, 36\rangle \\
\langle16, 2, 31, 7\rangle &
\end{array}$$

For the type $9^{10} 27^1$, let $X=I_{117}$, and ${\cal
G}=\{\{i,i+10,i+20,\dots,i+80\}:0\leq i <10\}\cup\{\{90,91,92,\dots,116\}\}$. Then $(X,{\cal
G},{\cal C})$ is a $[3,1]$-GDC$(6)$ of type $9^{10} 27^1$, where
${\cal C}$ is obtained by developing the following codewords under the
automorphism group $G=\langle(0\ \ 1\ \ 2\ \ \cdots\ \ 89)(90\ \ 91\ \ \cdots\ \ 98)$ $(99\ \
100\ \ \cdots\ \ 107)\rangle$ $(108\ \ 109\ \ \cdots\ \ 116)\rangle$.
$$\begin{array}{lllll}
\langle90, 5, 54, 26\rangle &
\langle90, 13, 12, 20\rangle &
\langle90, 33, 46, 25\rangle &
\langle99, 11, 62, 66\rangle &
\langle2, 26, 7, 90\rangle \\
\langle1, 39, 37, 65\rangle &
\langle99, 87, 58, 54\rangle &
\langle21, 36, 79, 99\rangle &
\langle108, 66, 83, 9\rangle &
\langle108, 34, 78, 23\rangle \\
\langle0, 23, 65, 68\rangle  &
\langle89, 5, 36, 108\rangle &
\langle14, 77, 86, 70\rangle &
\langle99, 32, 46, 43\rangle &
\langle108, 28, 40, 62\rangle
\end{array}$$

\end{proof}

\begin{lemma}
\label{GDC-27^t 9^1}
There exists a $[3,1]$-GDC$(6)$ of type $27^t 9^1$ with size $27t(3t-1)$ for each $t\in\{4,5,6\}$.
\end{lemma}

\begin{proof}
For each $t$, let $X_t=I_{27t+9}$, and ${\cal G}_t=\{\{i,i+t,i+2t,\dots,i+26t\}:0\leq i <t\}\cup\{\{27t,27t+1,\dots,27t+8\}\}$. Then $(X_t,{\cal G}_t,{\cal C}_t)$ is a $[3,1]$-GDC$(6)$ of type $27^{t} 9^1$, where ${\cal C}_t$ is obtained by developing the following codewords under the automorphism group $G=\langle(0\ \ 1\ \ 2\ \ \cdots\ \ 27t-1)(27t\ \ 27t+1\ \ \cdots\ \ 27t+8)\rangle$.

\noindent $t=4$:
$$\begin{array}{lllll}
\langle108, 6, 35, 57\rangle &
\langle108, 16, 90, 77\rangle &
\langle83, 81, 0, 6\rangle &
\langle12, 67, 5, 108\rangle &
\langle27, 64, 105, 10\rangle \\
\langle29, 78, 52, 99\rangle &
\langle108, 65, 64, 22\rangle &
\langle0, 5, 94, 43\rangle &
\langle12, 102, 9, 3\rangle &
\langle70, 20, 59, 37\rangle \\
\langle92, 19, 82, 61\rangle&
\end{array}$$
\noindent $t=5$:
$$\begin{array}{lllll}
\langle43, 0, 79, 42\rangle &
\langle0, 8, 91, 84\rangle &
\langle135, 49, 81, 98\rangle &
\langle135, 92, 30, 14\rangle &
\langle126, 58, 72, 109\rangle \\
\langle135, 52, 73, 24\rangle &
\langle26, 48, 72, 79\rangle &
\langle80, 77, 118, 11\rangle &
\langle29, 40, 31, 133\rangle &
\langle33, 59, 107, 135\rangle \\
\langle3, 21, 109, 37\rangle &
\langle22, 35, 93, 94\rangle &
\langle97, 124, 1, 130\rangle &
\langle25, 48, 44, 126\rangle
\end{array}$$
\noindent $t=6$:
$$\begin{array}{lllll}
\langle162, 26, 96, 119\rangle &
\langle16, 23, 123, 104\rangle &
\langle162, 59, 34, 144\rangle &
\langle0, 26, 89, 13\rangle &
\langle59, 39, 80, 24\rangle \\
\langle75, 146, 119, 16\rangle &
\langle78, 75, 118, 157\rangle &
\langle162, 111, 139, 82\rangle &
\langle147, 60, 2, 37\rangle &
\langle53, 54, 7, 130\rangle \\
\langle33, 161, 157, 52\rangle &
\langle94, 108, 103, 15\rangle &
\langle89, 134, 156, 145\rangle &
\langle37, 70, 135, 86\rangle &
\langle86, 18, 49, 33\rangle \\
\langle0, 50, 101, 130\rangle &
\langle83, 73, 81, 162\rangle &
\end{array}$$
\end{proof}

\begin{lemma}
\label{GDC-27^t 18^1}
There exists a $[3,1]$-GDC$(6)$ of type $27^t 18^1$ with size $27t(3t+1)$ for each $t\in\{4,6\}$.
\end{lemma}

\begin{proof}
For each $t$, let $X_t=I_{27t+18}$, and ${\cal G}_t=\{\{i,i+t,i+2t,\dots,i+26t\}:0\leq i <t\}\cup\{\{27t,27t+1,\dots,27t+17\}\}$. Then $(X_t,{\cal G}_t,{\cal C}_t)$ is a $[3,1]$-GDC$(6)$ of type $27^{t} 18^1$, where ${\cal C}_t$ is obtained by developing the following codewords under the automorphism group $G=\langle(0\ \ 1\ \ 2\ \ \cdots\ \ 27t-1)(27t\ \ 27t+1\ \ \cdots\ \ 27t+8) (27t+9\ \ 27t+10\ \ \cdots\ \ 27t+17)\rangle$.

\noindent $t=4$:
$$\begin{array}{lllll}
\langle108, 31, 86, 96\rangle &
\langle108, 46, 21, 7\rangle &
\langle108, 26, 9, 20\rangle &
\langle97, 47, 96, 108\rangle &
\langle37, 15, 28, 18\rangle \\
\langle117, 88, 55, 94\rangle &
\langle117, 35, 5, 20\rangle &
\langle54, 9, 27, 16\rangle &
\langle117, 12, 54, 33\rangle&
\langle84, 43, 38, 81\rangle \\
\langle31, 29, 60, 117\rangle &
\langle53, 6, 88, 107\rangle &
\langle0, 23, 74, 37\rangle &
\end{array}$$
\noindent $t=6$:
$$\begin{array}{llll}
\langle162, 22, 110, 0\rangle &
\langle162, 24, 158, 46\rangle &
\langle171, 110, 161, 69\rangle &
\langle18, 35, 82, 162\rangle \\
\langle171, 139, 30, 46\rangle &
\langle136, 135, 13, 120\rangle &
\langle171, 160, 149, 81\rangle &
\langle99, 42, 91, 171\rangle \\
\langle162, 79, 3, 134\rangle &
\langle39, 119, 114, 128\rangle &
\langle128, 95, 154, 141\rangle &
\langle118, 12, 83, 74\rangle \\
\langle143, 106, 39, 54\rangle &
\langle156, 136, 113, 86\rangle &
\langle131, 135, 160, 66\rangle &
\langle48, 50, 29, 129\rangle \\
\langle91, 84, 81, 125\rangle &
\langle60, 105, 137, 91\rangle &
\langle0, 38, 99, 65\rangle &
\end{array}$$
\end{proof}

\begin{lemma}
\label{GDC-36^6 27^1}
There exists a $[3,1]$-GDC$(6)$ of type $36^6 27^1$ with size $5616$.
\end{lemma}

\begin{proof}
Let $X=I_{243}$, and ${\cal G}=\{\{i,i+6,i+12,\dots,i+210\}:0\leq i <6\}\cup\{\{216,217,\dots,242\}\}$. Then $(X,{\cal G},{\cal C})$ is a $[3,1]$-GDC$(6)$ of type $36^6 27^1$, where ${\cal C}$ is obtained by developing the following codewords under the automorphism group $G=\langle(0\ \ 1\ \ 2\ \ \cdots\ \ 215)(216\ \ 217\ \ \cdots\ \ 224)$ $(225\ \ 226\ \ \cdots\ \ 233)$ $(234\ \ 235\ \ \cdots\ \ 242)\rangle$.
$$\begin{array}{llll}
\langle33, 34, 92, 41\rangle &
\langle14, 148, 61, 107\rangle &
\langle159, 204, 154, 61\rangle &
\langle184, 211, 162, 165\rangle \\
\langle199, 99, 0, 89\rangle &
\langle175, 80, 11, 190\rangle &
\langle134, 42, 173, 216\rangle &
\langle225, 183, 107, 145\rangle \\
\langle0, 25, 105, 119\rangle &
\langle159, 96, 73, 116\rangle &
\langle225, 160, 23, 186\rangle &
\langle175, 105, 208, 225\rangle \\
\langle152, 27, 214, 7\rangle &
\langle34, 135, 146, 131\rangle &
\langle225, 153, 92, 157\rangle &
\langle234, 118, 149, 159\rangle \\
\langle216, 34, 91, 26\rangle &
\langle170, 103, 29, 156\rangle &
\langle234, 124, 207, 89\rangle &
\langle166, 134, 132, 234\rangle \\
\langle216, 14, 54, 51\rangle &
\langle142, 198, 121, 33\rangle &
\langle234, 200, 48, 193\rangle &
\langle181, 209, 165, 200\rangle \\
\langle7, 20, 155, 129\rangle &
\langle216, 129, 184, 56\rangle &
\end{array}$$
\end{proof}

\begin{lemma}
\label{GDCa} There is a $[3,1]$-GDC$(6)$ of type $36^u(9x)^1$ with size $72u(2u+x-2)$ for $u\geq 4$, $u\not\in \{6,10\}$ and $0\leq x\leq u$.
\end{lemma}

\begin{proof}
Take a TD$(5,u)$ with $u\geq 4$, $u\not\in \{6,10\}$ from
Theorem~\ref{TD}; remove one point from one group and adjoin an ideal point to obtain a
$\{5,u+1\}$-GDD of type $4^u u^1$. Apply the Fundamental
Construction with weight $9$ to the points in the groups of
size $4$ and $x$ points in the group of size $u$ for $0\leq
x\leq u$. Note that there exist $[3,1]$-GDC$(6)$s of type
$9^s$ for $s\geq 4$ by Theorem~\ref{9^t}. The result is the
desired GDC.
\end{proof}

\begin{lemma}
\label{GDCb} The following $[3,1]$-GDC$(6)$s all exist:
\begin{enumerate}
\item[i)] type $45^5(9x)^1$ and size  $450(10+x)$ for $x\in \{0,1,2\}$;
\item[ii)] type $36^6 (9x)^1$ and size $432(10+x)$ for $x\in \{0,1,3\}$;
\item[iii)] type $36^{10} (9x)^1$ and size $720(18+x)$ for $x\in \{0,1,2,3\}$;
\item[iv)] type $27^t$ and size $81t(t-1)$ for $t\in \{4,5,7\}$;
\item[v)] type $18^7$ and size $1512$.
\end{enumerate}
\end{lemma}

\begin{proof}
For i), take a TD$(6,5)$ from Theorem~\ref{TD}, and apply the
Fundamental Construction with weight $9$ to all the points in
the first five groups and $x$ points in the last group.
For ii), take a $[3,1]$-GDC$(6)$ of type $9^6$ from Theorem~\ref{9^t}; apply the Inflation Construction with weight $4$, using $4$-MGDDs of type $4^4$ as input designs, to obtain a $4$-DGDD of type $(36,9^4)^6$ with the CCC property. Then adjoin $9x$ ideal points with $x\in\{0,1\}$ and fill in $[3,1]$-GDC$(6)$s of type $9^6 (9x)^1$ to obtain a $[3,1]$-GDC$(6)$ of type $36^6 (9x)^1$, as desired; for $x=3$, see Lemma~\ref{GDC-36^6 27^1}.
For iii), we proceed similarly, starting instead with a $[3,1]$-GDC$(6)$ of type $9^{10}$ to obtain a $4$-DGDD of type $(36,9^4)^{10}$ with the CCC property. Then adjoin $9x$ ideal points and fill in $[3,1]$-GDC$(6)$s of type $9^{10} (9x)^1$ from Theorem~\ref{9^t} and Lemma~\ref{GDC-9^10(18&27)^1} to obtain the desired GDCs.
For iv), inflate a
$[3,1]$-GDC$(6)$ of type $9^t$ from Theorem~\ref{9^t} with weight $3$.
For v), apply the Fundamental
Construction with weight $9$ to a $4$-GDD of type $2^7$ (see \cite{Ge:2007}) to
obtain the desired GDC.
\end{proof}

\subsection{Cases of Length $n\equiv 0,1,2,3\pmod{9}$}

\begin{lemma}
\label{s1(9)} $A_3(9t+1,6,[3,1])=U(9t+1,6,[3,1])$ for $t\in
\{1,3\}$.
\end{lemma}

\begin{proof}
For $t=1$, see Table~\ref{presult}. For $t=3$, let $X=\bbZ_{28}$. Then $(X,{\cal
C})$ is the desired optimal $(28,6,[3,1])_3$-code, where ${\cal C}$ is obtained by developing the following base codewords $+4\pmod{28}$.
$$\begin{array}{llllll}
\langle1,3,26,8\rangle & \langle2,8,18,26\rangle & \langle2,16,17,22\rangle &
\langle3,16,18,7\rangle & \langle0,8,25,21\rangle & \langle1,21,22,16\rangle \\
\langle2,3,23,6\rangle & \langle0,3,24,12\rangle & \langle1,17,27,23\rangle &
\langle0,9,26,23\rangle & \langle3,19,20,25\rangle & \langle3,13,22,17\rangle \\
\end{array}$$
\end{proof}

\begin{lemma}
\label{s0(9)} $A_3(9t,6,[3,1])=U(9t,6,[3,1])$ for each $t\in
\{1,2,3\}$.
\end{lemma}

\begin{proof}
For each $t\in \{1,3\}$, the desired code can be shorten from an
optimal $(9t+1,6,[3,1])_3$-code.
For $t=2$, let $X=\bbZ_{18}$. Then an
$(18,6,[3,1])_3$-code with $30$ codewords is obtained by
developing the codewords $\langle0,1,2,3\rangle$, $\langle8,11,13,2\rangle$, $\langle0,7,10,16\rangle$, $\langle0,5,12,9\rangle$, and $\langle2,12,16,8\rangle$ $+3\pmod{18}$.
\end{proof}

\begin{lemma}
\label{s1^9m^1} There exists a $[3,1]$-GDC$(6)$ of type $1^9 2^1$
with size $11$ and a  $[3,1]$-GDC$(6)$ of type $1^9 3^1$
with size $12$.
\end{lemma}

\begin{proof}
The codewords for each desired code constructed on $I_{11}$ or $I_{12}$ are listed below.

\noindent $1^9 2^1:$
$$\begin{array}{llllll}
\langle1,3,9,0\rangle & \langle1,4,7,8\rangle & \langle2,8,9,4\rangle & \langle5,8,10,1\rangle & \langle0,4,5,9\rangle &  \langle1,2,6,10\rangle\\ \langle0,7,10,2\rangle & \langle2,3,5,7\rangle &
\langle0,6,8,3\rangle & \langle3,4,10,6\rangle & \langle6,7,9,5\rangle \\
\end{array}$$

\noindent $1^9 3^1:$
$$\begin{array}{llllll}
\langle3,7,11,0\rangle & \langle2,4,7,10\rangle & \langle7,8,9,1\rangle & \langle0,2,9,3\rangle & \langle0,1,10,7\rangle & \langle6,8,10,4\rangle \\ \langle0,5,8,11\rangle & \langle1,4,11,8\rangle & \langle1,3,6,9\rangle & \langle4,5,9,6\rangle & \langle2,6,11,5\rangle & \langle3,5,10,2\rangle \\
\end{array}$$
\end{proof}

\begin{lemma}
\label{s2(9)} $A_3(9t+2,6,[3,1])=U(9t+2,6,[3,1])$ for $t\in
\{1,2,3\}$.
\end{lemma}

\begin{proof}
For $t=1$, the $[3,1]$-GDC of type $1^92^1$ constructed in
Lemma~\ref{s1^9m^1} is the desired code. For each $t\in \{2,3\}$,
let $X_t=\bbZ_{9t+2}$. Then $(X_t,{\cal C}_t)$ is the desired
optimal $(9t+2,6,[3,1])_3$-code, where ${\cal C}_t$ is  obtained by developing the elements of $\bbZ_{9t+2}$ in
the following codewords $+1\pmod{9t+2}$.

\noindent $t=2$: $\langle0,1,9,15\rangle$
$\langle0,3,7,5\rangle$

\noindent $t=3$: $\langle0,1,26,15\rangle$
$\langle0,6,13,11\rangle$ $\langle0,8,20,10\rangle$
\end{proof}

\begin{lemma}
\label{s3(9)} $A_3(9t+3,6,[3,1])=U(9t+3,6,[3,1])$ for each
$t\in \{1,2,3\}$.
\end{lemma}

\begin{proof}
For $t=1$, the $[3,1]$-GDC of type $1^93^1$ constructed in
Lemma~\ref{s1^9m^1} is the desired code. For $t=2$, the
$[3,1]$-GDC$(6)$ of type $3^7$ constructed in Lemma~\ref{s3^3t+1} is the
desired code. For $t=3$, inflate a $[3,1]$-GDC$(6)$ of type
$1^{10}$ with weight $3$ to obtain a $[3,1]$-GDC$(6)$ of type
$3^{10}$, which is the desired code.
\end{proof}

\begin{theorem}
$A_3(9t+i,6,[3,1])=U(9t+i,6,[3,1])$ for each $t\geq 1$ and
$i\in \{0,1,2,3\}$, except possibly for $(t,i)=(2,1)$.
\end{theorem}

\begin{proof} For $t\leq 3$, see Lemmas~\ref{s1(9)}--\ref{s0(9)}, \ref{s2(9)} and \ref{s3(9)}.
For each $t\geq 4$, take a $[3,1]$-GDC$(6)$ of type $9^t$ constructed
in Theorem~\ref{9^t}. Adjoin $i$ ideal points and fill in the groups
with $[3,1]$-GDC$(6)$s of type $1^{9}i^1$ constructed in
Lemmas~\ref{s1(9)}, \ref{s0(9)} and \ref{s1^9m^1} to obtain the
desired optimal $(9t+i,6,[3,1])_3$-code.
\end{proof}

\subsection{Case of Length $n\equiv 6\pmod{9}$}

\begin{lemma}
\label{s1^9t6^1} There is a $[3,1]$-GDC$(6)$ of type $1^{9t}6^1$
with size $9t(t+1)$ for each $t\in \{2,3,5,7,9\}$.
\end{lemma}

\begin{proof}
Let $X_t=I_{9t+6}$, and ${\cal G}_t=\{\{i\}:0\leq
i\leq 9t-1\}\cup\{\{9t,9t+1,\dots,9t+5\}\}$. Then $(X_t,{\cal
G}_t,{\cal C}_t)$ is a $[3,1]$-GDC$(6)$ of type $1^{9t}6^1$,
where ${\cal C}_t$ is obtained by developing the codewords in Table~\ref{t1^9t6^1} under the automorphism group
$G_t$.

For $t\in\{2,3\}$, $G_t=\langle(0\ \ 3\ \ 6\ \ \cdots\ \
9t-3)(1\ \ 4\ \ 7\ \ \cdots\ \ 9t-2)(2\ \ 5\ \ 8\ \ \cdots\ \
9t-1)(9t\ \ 9t+1\ \ 9t+2)(9t+3\ \ 9t+4\ \ 9t+5)\rangle$.

For $t\in\{5,7,9\}$, $G_t=\langle(0\ \ 1\ \ 2\ \ \cdots\ \ 9t-1)(9t\
\ 9t+1\ \ 9t+2)(9t+3\ \ 9t+4\ \ 9t+5)\rangle$.
\end{proof}

\begin{table*}
\centering
\renewcommand{\arraystretch}{1}
\setlength\arraycolsep{3pt} \caption{Base Codewords of Small
$[3,1]$-GDC$(6)$s of type $1^{9t}6^1$ in Lemma~\ref{s1^9t6^1}}
\label{t1^9t6^1}
\begin{tabular}{c|l}
\hline $t$ & {\hfill Codewords \hfill} \\
\hline $2$ & $\begin{array}{lllllll}
\langle0,1,19,2\rangle & \langle1,3,22,16\rangle & \langle0,7,23,14\rangle & \langle2,21,5,6\rangle & \langle1,9,20,6\rangle & \langle2,18,14,1\rangle & \langle0,8,4,17\rangle \\ \langle0,6,11,9\rangle & \langle2,10,4,13\rangle \\
\end{array}$ \\
\hline $3$ & $\begin{array}{lllllll}
\langle1,6,17,30\rangle & \langle1,31,14,9\rangle & \langle0,20,32,16\rangle & \langle0,29,25,5\rangle & \langle1,29,22,15\rangle & \langle2,11,1,12\rangle \\
\langle2,23,4,26\rangle & \langle0,21,18,17\rangle & \langle2,17,29,21\rangle & \langle2,15,0,7\rangle & \langle0,4,30,8\rangle & \langle0,1,10,13\rangle \\
\end{array}$ \\
\hline $5$ & $\begin{array}{lllllll}
\langle0,39,15,20\rangle & \langle0,35,36,33\rangle & \langle0,13,41,8\rangle & \langle0,31,49,38\rangle & \langle0,47,23,25\rangle & \langle0,11,29,3\rangle \\
\end{array}$ \\
\hline $7$ & $\begin{array}{llllll}
\langle0,5,21,53\rangle & \langle0,20,65,7\rangle & \langle0,25,54,31\rangle & \langle0,2,3,15\rangle & \langle0,55,68,44\rangle & \langle0,33,37,56\rangle \\ \langle0,27,41,51\rangle & \langle0,17,35,11\rangle \\
\end{array}$ \\
\hline $9$ & $\begin{array}{lllllll}
\langle0,19,58,28\rangle & \langle0,6,32,44\rangle & \langle0,33,50,30\rangle & \langle0,41,85,52\rangle & \langle0,4,67,20\rangle & \langle0,8,82,7\rangle \\
\langle0,2,15,3\rangle & \langle0,25,35,72\rangle & \langle0,5,27,70\rangle & \langle0,21,45,74\rangle &
\end{array}$ \\
\hline
\end{tabular}
\end{table*}

\begin{lemma}
\label{s6(9)} $A_3(9t+6,6,[3,1])=U(9t+6,6,[3,1])$ for $0\leq t\leq
10$.
\end{lemma}

\begin{proof}
For $t=0$, see Table~\ref{presult}. For $t=1$, let $X=\bbZ_{15}$. Then $(X,{\cal
C})$ is the desired optimal $(15,6,[3,1])_3$-code, where
${\cal C}$ is obtained by developing the codewords $\langle2,5,0,7\rangle$,
$\langle1,8,2,9\rangle$, $\langle0,4,9,8\rangle$, and $\langle1,13,0,11\rangle$ $+3\pmod{15}$.
For each $t\in \{4,6,8,10\}$, take a $[3,1]$-GDC$(6)$ of type
$6^{3t/2+1}$ from Lemma~\ref{s6^t}, and fill in the groups with
optimal $(6,6,[3,1])_3$-codes to obtain the desired code.
For each $t\in \{2,3,5,7,9\}$, take a $[3,1]$-GDC$(6)$ of type
$1^{9t}6^1$ from Lemma~\ref{s1^9t6^1}, and fill in the group of size
$6$ with an optimal $(6,6,[3,1])_3$-code to obtain the desired code.
\end{proof}

\begin{lemma}
\label{6(9)}$A_3(9t+6,6,[3,1])=U(9t+6,6,[3,1])$ for $t\geq 11$.
\end{lemma}

\begin{proof}
For each $t\geq 16$ and $t\not=26$, write $t=4u+x$ with $u\geq 4$ and $x\in\{0,1,2,3\}$. Take a
$[3,1]$-GDC$(6)$ of type $36^u(9x)^1$ from Lemmas~\ref{GDCa}--\ref{GDCb}.
Adjoin six ideal points, and fill in the groups of size $36$
together with the ideal points with $[3,1]$-GDC$(6)$s of type
$6^7$ from Lemma~\ref{s6^t} to obtain a $[3,1]$-GDC$(6)$ of
type $6^{6u}(9x+6)^1$. Then fill in the groups with optimal
$(6,6,[3,1])_3$-codes and an optimal $(9x+6,6,[3,1])_3$-code
from Lemma~\ref{s6(9)} to obtain an optimal
$(36u+9x+6,6,[3,1])_3$-code, as desired.
For $t=26$, take a $[3,1]$-GDC$(6)$ of type
$45^5 9^1$ from Lemma~\ref{GDCb}.
Adjoin $6$ ideal points, fill in the groups of size $45$
together with these ideal points with $[3,1]$-GDC$(6)$s of type
$1^{45}6^1$ from Lemma~\ref{s1^9t6^1}, and fill in the group of size
$9$ together with these ideal points with an optimal
$(15,6,[3,1])_3$-code to obtain the desired code.

For each $t\in \{12,14,15\}$, take a $[3,1]$-GDC$(6)$ of type
$27^4$, type $18^7$, or type $27^5$ from Lemma~\ref{GDCb}. Adjoin
$6$ ideal points and fill in the groups to obtain the desired code.
For $t=13$, we proceed similarly, starting instead with a
$[3,1]$-GDC$(6)$ of type $27^4 9^1$ from Lemma~\ref{GDC-27^t 9^1}.

For $t=11$, inflate a $[3,1]$-GDC$(6)$ of type $3^7$ from
Lemma~\ref{s3^3t+1} with weight $5$ to obtain a $[3,1]$-GDC$(6)$ of
type $15^7$. Fill in the groups with optimal $(15,6,[3,1])_3$-codes to complete the proof.
\end{proof}

Combining Lemmas~\ref{s6(9)} and \ref{6(9)}, we have the
following result.

\begin{theorem}
$A_3(9t+6,6,[3,1])=U(9t+6,6,[3,1])$ for each $t\geq 0$.
\end{theorem}

\subsection{Case of Length $n\equiv 4\pmod{9}$}

\begin{lemma}
\label{s4^t} There exists a $[3,1]$-GDC$(6)$ of type $4^{9t+1}$ with size $16t(9t+1)$
for each $t\in \{1,2,3\}$.
\end{lemma}

\begin{proof}
For $t=2$, let $X=\bbZ_{76}$, and ${\cal
G}=\{\{0,19,38,57\}+i:0\leq i\leq 18\}$. Then $(X,{\cal
G},{\cal C})$ is a $[3,1]$-GDC$(6)$ of type $4^{19}$, where
${\cal C}$ is obtained by developing the elements of $\bbZ_{76}$
in the following codewords $+1\pmod{76}$.
$$\begin{array}{llll}
\langle0,26,22,9\rangle &
\langle0,62,74,13\rangle &
\langle0,37,6,67\rangle &
\langle0,44,65,41\rangle \\
\langle0,5,56,29\rangle &
\langle0,43,53,46\rangle &
\langle0,42,60,1\rangle &
\langle0,8,36,7\rangle
\end{array}$$

For each $t\in \{1,3\}$, inflate a $[3,1]$-GDC$(6)$ of type
$1^{9t+1}$ from Lemma~\ref{s1(9)} with weight $4$ to obtain the desired code.
\end{proof}

\begin{lemma}
\label{s4(9)} $A_3(9t+4,6,[3,1])=U(9t+4,6,[3,1])$ for $t\in
\{0,2,3,4,8,12\}$.
\end{lemma}

\begin{proof}
For $t=0$, see Table~\ref{presult}.

For $t\in \{2,3\}$, let $X_t=\bbZ_{9t+3}\cup\{\infty\}$. Then
$(X_t,{\cal C}_t)$ is the desired optimal $(9t+4,6,[3,1])_3$-code, where ${\cal C}_t$ is developed $+3\pmod{9t+3}$.

\noindent $t=2$: $\langle1,6,12,5\rangle$
$\langle0,12,8,19\rangle$ $\langle1,8,18,15\rangle$
$\langle0,1,2,3\rangle$ $\langle2,5,11,7\rangle$
$\langle1,10,7,20\rangle$ $\langle0,5,\infty,13\rangle$

\noindent $t=3$: $\langle0,2,16,6\rangle$
$\langle0,1,25,23\rangle$ $\langle0,7,28,15\rangle$
$\langle2,5,22,4\rangle$ $\langle2,7,18,26\rangle$
$\langle2,21,11,17\rangle$ $\langle0,9,12,22\rangle$
$\langle2,28,25,13\rangle$ $\langle0,17,29,24\rangle$
$\langle0,\infty,4,5\rangle$

For $t\in \{4,8,12\}$, take $[3,1]$-GDC$(6)$s of type
$4^{9t/4+1}$ from Lemma~\ref{s4^t}, and fill in the groups with
optimal $(4,6,[3,1])_3$-codes to obtain the desired codes.
\end{proof}

\begin{lemma}
$A_3(9t+4,6,[3,1])=U(9t+4,6,[3,1])$ for each $t\not\equiv
1\pmod{4}$, $t\geq 16$ and $t\not=26$.
\end{lemma}

\begin{proof}
For each $t\not\equiv 1\pmod{4}$, $t\geq 16$ and $t\not=26$, write $t=4u+x$ such that $x\in\{0,2,3\}$. Take a $[3,1]$-GDC$(6)$ of type $36^u(9x)^1$ from Lemmas~\ref{GDCa}--\ref{GDCb}, adjoin four ideal points, and fill in the groups of size $36$ together with these ideal points with $[3,1]$-GDC$(6)$s of type $4^{10}$ from Lemma~\ref{s4^t} to
obtain a $[3,1]$-GDC$(6)$ of type $4^{9u}(9x+4)^1$. Fill in the
groups with optimal $(4,6,[3,1])_3$-codes and an optimal
$(9x+4,6,[3,1])_3$-code from Lemma~\ref{s4(9)} to obtain an
optimal $(36u+9x+4,6,[3,1])_3$-code, as desired.
\end{proof}

\begin{lemma}
$A_3(9t+4,6,[3,1])=U(9t+4,6,[3,1])$ for $t\in\{17,33\}$ or $t\equiv 1\pmod{4}, t\geq 85$.
\end{lemma}

\begin{proof}
For $t=17$, take a $[3,1]$-GDC$(6)$ of type $39^4$ from Lemma~\ref{GDC-39^4} and adjoin an ideal point. Then fill in the groups together with the ideal point with optimal codes of length $40$ from Lemma~\ref{s4(9)} to obtain the desired code.

For $t=33$, take a TD$(4,5)$ and assign weight $15$ to each point of this TD. Note that there exists a $[3,1]$-GDC$(6)$ of type $15^4$ by Lemma~\ref{s15^t}. We obtain a $[3,1]$-GDC$(6)$ of type $75^4$. Now adjoin an ideal point and fill in the groups together with the ideal point with optimal codes of length $76$ to obtain the desired code.

For $t\geq 85$, write $t=4u+17$ with $u \geq 17$. Take a $[3,1]$-GDC$(6)$ of type $36^u 153^1$ from Lemma~\ref{GDCa}, adjoin four ideal points, and fill in the groups to obtain an
optimal $(36u+157,6,[3,1])_3$-code.
\end{proof}

Summarizing the above results, we have:

\begin{theorem}
\label{4(9)} $A_3(9t+4,6,[3,1])=U(9t+4,6,[3,1])$ for each
$t\geq 0$ and $t\not\in \{1,5,6,7,9,10,$ $11,13,14,15,21,25,26,29,37,41,45,49,53,57,61,65,69,73,77,81\}$.
\end{theorem}

\subsection{Case of Length $n\equiv 5\pmod{9}$}

\begin{lemma}
\label{GDC-1^36 5^1}
There exists a $[3,1]$-GDC$(6)$ of type $1^{36} 5^1$ with 173 codewords.
\end{lemma}

\begin{proof}
Let $X=(\bbZ_{12} \times \{0, 1, 2\})\cup (\bbZ_{3} \times \{3\}) \cup \{\infty_0, \infty_1\}$. The point set $(\bbZ_{3} \times \{3\}) \cup \{\infty_0, \infty_1\}$ forms the group of size 5. Define $\alpha: X \rightarrow X$ as $x_y \rightarrow (x+1)_y$ where the addition is modulo 12 if $y \in \{0, 1, 2\}$, and modulo 3 if $y =3$. The points $\infty_0$ and $\infty_1$ are fixed by $\alpha$.
Develop the following 13 base codewords with $\alpha$:
$$\begin{array}{lllll}
\langle0_3, 6_0, 2_2, 0_2\rangle &
\langle3_2, 11_1, 7_2, 0_3\rangle &
\langle10_2, 10_1, 3_2, 1_0\rangle &
\langle4_2, 0_0, 9_0, 7_1\rangle &
\langle8_0, 10_0, 8_2, 7_2\rangle \\
\langle1_0, 0_1, 2_2, 3_2\rangle &
\langle0_3, 4_0, 11_0, 8_1\rangle  &
\langle1_1, 11_1, 10_0, 9_0\rangle &
\langle0_1, 1_2, 7_1, 2_0\rangle &
\langle\infty_0, 2_1, 1_2, 7_0\rangle \\
\langle0_3, 3_1, 4_1, 1_2\rangle &
\langle\infty_1, 9_2, 4_0, 6_1\rangle &
\langle1_1, 5_0, 8_2, 10_1\rangle &
\end{array}$$
Further develop the codeword $\langle0_0, 6_0, 0_1, 6_1\rangle$ making a short orbit of length 6.\\
Finally develop the following codewords making 2 short orbits of length 4 and one short orbit of length 3:
$$\begin{array}{lll}
\langle0_0, 4_0, 8_0, \infty_0\rangle &
\langle0_1, 4_1, 8_1, \infty_1\rangle &
\langle0_2, 3_2, 6_2, 9_2\rangle
\end{array}$$
\end{proof}

\begin{lemma}
\label{s5(9)} $A_3(9t+5,6,[3,1])=U(9t+5,6,[3,1])$ for $t\in
\{0,1,2,3,4\}$.
\end{lemma}

\begin{proof}
For $t=0$, see Table~\ref{presult}.

For $t=1$, an optimal $(14,6,[3,1])_3$-code is constructed on
$I_{14}$ with $17$ codewords listed below.
$$\begin{array}{llllll}
\langle0,4,9,5\rangle & \langle7,4,10,6\rangle & \langle2,10,9,3\rangle & \langle13,10,0,8\rangle &
\langle2,12,8,4\rangle & \langle0,1,7,12\rangle \\ \langle3,5,8,0\rangle & \langle2,5,7,13\rangle &
\langle6,9,8,11\rangle & \langle3,6,12,10\rangle & \langle3,13,4,2\rangle & \langle12,13,9,7\rangle \\
\langle2,11,0,6\rangle & \langle11,4,1,8\rangle & \langle11,3,7,9\rangle & \langle11,10,5,12\rangle &
\langle13,1,6,5\rangle &  \\
\end{array}$$

For $t=2$, the desired code is constructed on
$\bbZ_{21}\cup\{\infty_1,\infty_2\}$ and is obtained by developing the
elements of $\bbZ_{21}$ in the following codewords $+7\pmod{21}$, where the infinite points $\infty_0, \infty_1$ are fixed when developed.
$$\begin{array}{llllll}
\langle0,1,2,3\rangle &
\langle9,20,16,3\rangle &
\langle7,16,11,19\rangle &
\langle5,8,19,2\rangle &
\langle2,12,7,20\rangle &
\langle0,11,18,10\rangle \\
\langle10,5,4,6\rangle &
\langle15,3,10,2\rangle &
\langle1,14,20,\infty_2\rangle &
\langle6,7,3,14\rangle &
\langle2,\infty_2,4,15\rangle &
\langle9,\infty_1,18,15\rangle \\
\langle1,7,5,\infty_1\rangle &
\langle\infty_2,12,3,0\rangle &
\langle10,\infty_1,20,12\rangle &
\langle8,1,13,11\rangle &
\langle13,18,6,12\rangle &
\end{array}$$

For $t=3$, the desired code is constructed on $\bbZ_{32}$ and is obtained by developing the
elements of $\bbZ_{32}$ in the following codewords $+4\pmod{32}$.
$$\begin{array}{llllll}
\langle2,8,31,1\rangle & \langle2,9,24,23\rangle & \langle2,14,17,30\rangle &
\langle1,12,23,2\rangle & \langle3,20,22,26\rangle & \langle0,18,29,13\rangle \\
\langle0,28,8,3\rangle & \langle2,26,27,28\rangle & \langle2,20,25,21\rangle &
\langle3,1,27,31\rangle & \langle3,23,18,28\rangle & \langle3,17,29,21\rangle \\
\langle0,9,19,16\rangle &
\end{array}$$

Finally, for $t=4$, take the $[3,1]$-GDC$(6)$ of type $1^{36} 5^1$ constructed in Lemma~\ref{GDC-1^36 5^1} and fill in the group of size 5 with an optimal code of length 5.
\end{proof}

\begin{lemma}
$A_3(9t+5,6,[3,1])=U(9t+5,6,[3,1])$ for each
$t\geq 16$ and $t\not=26$.
\end{lemma}

\begin{proof}
For each $t\geq 16$ and $t\not=26$, write $t=4u+x$ such that $0\leq x\leq 3$. Take a $[3,1]$-GDC$(6)$ of type $36^u(9x)^1$ from Lemmas~\ref{GDCa} and~\ref{GDCb}. Adjoin five ideal points, and fill in the groups of size $36$ together with the ideal points with $[3,1]$-GDC$(6)$s of type
$1^{36} 5^1$ from Lemma~\ref{GDC-1^36 5^1} to obtain a $[3,1]$-GDC$(6)$ of type $1^{36u}(9x+5)^1$. Then fill in the group with an optimal code of length $9x+5$ from Lemma~\ref{s5(9)} to obtain an optimal
code of length $9(4u+x)+5$, as desired.
\end{proof}

\begin{theorem}
\label{5(9)} $A_3(9t+5,6,[3,1])=U(9t+5,6,[3,1])$ for each
$t\geq 0$ and $t\not\in [5,15]\cup \{26\}$.
\end{theorem}

\subsection{Case of Length $n\equiv 7\pmod{9}$}

\begin{lemma}
\label{GDC-1^36 7^1}
There exists a $[3,1]$-GDC$(6)$ of type $1^{36} 7^1$ with 190 codewords.
\end{lemma}

\begin{proof}
Let $X=(\bbZ_{12} \times \{0, 1, 2\})\cup (\bbZ_{3} \times \{3,4\}) \cup \{\infty\}$. The point set $(\bbZ_{3} \times \{3,4\}) \cup \{\infty\}$ forms the group of size 7. Define $\alpha: X \rightarrow X$ as $x_y \rightarrow (x+1)_y$ where the addition is modulo 12 if $y \in \{0, 1, 2\}$, and modulo 3 if $y \in \{3,4\}$. The point $\infty$ is fixed by $\alpha$.
Develop the following 15 base codewords with $\alpha$:
$$\begin{array}{lllll}
\langle0_3, 4_2, 6_1, 8_0\rangle &
\langle0_3, 9_2, 4_1, 8_1\rangle &
\langle0_3, 5_2, 7_0, 0_0\rangle &
\langle8_0, 6_0, 11_1, 0_3\rangle &
\langle0_4, 6_1, 11_1, 6_2\rangle \\
\langle0_4, 0_0, 1_1, 7_2\rangle &
\langle0_4, 8_2, 7_0, 8_0\rangle &
\langle0_2, 8_0, 7_2, 0_4\rangle &
\langle0_0, 3_0, 11_1, 3_2\rangle &
\langle3_0, 0_2, 10_1, 1_1\rangle \\
\langle8_1, 4_0, 6_1, 9_2\rangle &
\langle6_0, 2_0, 8_2, 1_0\rangle &
\langle\infty, 4_0, 1_1, 0_2\rangle &
\langle8_2, 10_2, 7_2, 4_1\rangle &
\langle10_1, 7_2, 11_1, 7_1\rangle \\
\end{array}$$
Further develop the following codewords making a short orbit of length 6 and a short orbit of length 4:
$$\begin{array}{ll}
\langle0_0, 6_0, 0_1, 6_1\rangle &
\langle0_2, 4_2, 8_2, \infty\rangle \\
\end{array}$$
\end{proof}

\begin{lemma}
\label{GDC-1^54 7^1}
There exists a $[3,1]$-GDC$(6)$ of type $1^{54} 7^1$ with 393 codewords.
\end{lemma}

\begin{proof}
Let $X=(\bbZ_{18} \times \{0, 1, 2\})\cup (\bbZ_{3} \times \{3,4\}) \cup \{\infty\}$. The point set $(\bbZ_{3} \times \{3,4\}) \cup \{\infty\}$ forms the group of size 7. Define $\alpha: X \rightarrow X$ as $x_y \rightarrow (x+1)_y$ where the addition is modulo 18 if $y \in \{0, 1, 2\}$, and modulo 3 if $y \in \{3,4\}$. The point $\infty$ is fixed by $\alpha$.
Develop the following 21 base codewords with $\alpha$:
$$\begin{array}{lllll}
\langle3_1, 5_2, 8_2, 4_0\rangle &
\langle0_4, 1_2, 8_0, 15_2\rangle &
\langle0_3, 17_2, 8_1, 10_0\rangle &
\langle0_0, 2_1, 13_2, 12_1\rangle &
\langle0_4, 3_0, 13_0, 11_2\rangle \\
\langle3_0, 5_0, 9_1, 7_2\rangle &
\langle14_2, 1_2, 2_0, 5_2\rangle &
\langle1_2, 2_2, 13_1, 17_0\rangle &
\langle0_3, 4_1, 11_0, 12_2\rangle &
\langle11_0, 12_0, 5_0, 10_1\rangle \\
\langle0_0, 5_0, 8_1, 1_1\rangle &
\langle0_5, 7_0, 4_1, 17_2\rangle &
\langle6_0, 3_2, 11_2, 13_1\rangle &
\langle0_3, 13_2, 9_1, 15_0\rangle &
\langle1_1, 13_2, 14_1, 15_2\rangle \\
\langle6_1, 3_1, 3_2, 11_0\rangle &
\langle0_4, 10_1, 9_1, 17_1\rangle &
\langle3_0, 12_2, 6_0, 16_1\rangle &
\langle15_2, 15_0, 1_0, 0_3\rangle &
\langle0_1, 12_1, 16_1, 11_0\rangle \\
\langle8_1, 0_2, 11_2, 0_4\rangle &
\end{array}$$
Further develop the following codewords making a short orbit of length 9 and a short orbit of length 6:
$$\begin{array}{ll}
\langle0_0, 9_0, 0_1, 9_1\rangle &
\langle0_2, 6_2, 12_2, 0_5\rangle
\end{array}$$
\end{proof}

\begin{lemma}
\label{s7(9)} $A_3(7,6,[3,1])=2$, $A_3(9t+7,6,[3,1])=U(9t+7,6,[3,1])$ for each $t\in
\{1,2,3,6,24\}$.
\end{lemma}

\begin{proof} For $t=0$, see Table~\ref{presult}.

For $t=1$, the desired code is constructed on
$I_{16}$ with $24$ codewords listed below.
$$\begin{array}{llll}
\langle4,6,9,5\rangle & \langle9,0,3,13\rangle & \langle5,2,10,11\rangle & \langle9,15,2,14\rangle \\
\langle1,0,6,8\rangle & \langle4,8,14,2\rangle & \langle15,6,12,7\rangle & \langle11,0,14,5\rangle \\
\langle8,5,7,6\rangle & \langle0,12,2,4\rangle & \langle10,7,0,15\rangle & \langle8,13,15,0\rangle \\
\langle3,2,1,7\rangle & \langle3,5,15,4\rangle & \langle10,6,3,14\rangle & \langle11,4,1,15\rangle \\
\langle1,5,13,9\rangle & \langle11,7,9,12\rangle & \langle14,7,13,3\rangle & \langle13,10,4,12\rangle \\
\langle9,10,8,1\rangle & \langle3,8,12,11\rangle & \langle13,11,6,2\rangle & \langle1,12,14,10\rangle \\
\end{array}$$

For $t=2$, the desired code is constructed on $\bbZ_{24}\cup\{\infty\}$. Develop the following codewords $+4\pmod{24}$.
$$\begin{array}{llllll}
\langle0, 9, 10, 5\rangle &
\langle21, 5, 18, 9\rangle &
\langle10, 17, 4, 22\rangle &
\langle7, 16, 18, 11\rangle &
\langle11, 17, 0, 12\rangle \\
\langle6, 7, 2, 12\rangle &
\langle9, 8, 12, 11\rangle &
\langle1, 23, 15, 11\rangle &
\langle0, 7, 22, 14\rangle &
\langle\infty, 23, 17, 2\rangle
\end{array}$$
Then add two codewords  $\langle0, 8, 16, \infty\rangle$ and $\langle4, 12, 20, \infty\rangle$.

For $t=3$, the desired code is constructed on $X=(\bbZ_{6} \times \{0, 1, 2, 3, 4\})\cup (\bbZ_{2} \times \{5,6\})$. Define $\alpha: X \rightarrow X$ as $x_y \rightarrow (x+1)_y$ where the addition is modulo $6$ if $y \in \{0, 1, 2, 3, 4\}$, and modulo $2$ if $y \in \{5, 6\}$.
Develop the following $19$ base codewords with $\alpha$:
$$\begin{array}{llllll}
\langle0_5, 1_4, 1_1, 0_2\rangle &
\langle0_5, 3_0, 2_4, 5_2\rangle &
\langle0_5, 5_3, 2_1, 4_3\rangle &
\langle0_2, 5_2, 2_3, 0_5\rangle &
\langle0_6, 3_1, 2_0, 1_3\rangle &
\langle0_6, 0_3, 2_2, 5_4\rangle \\
\langle0_6, 0_4, 5_0, 2_1\rangle &
\langle3_2, 5_4, 3_3, 0_6\rangle &
\langle1_2, 0_1, 4_0, 3_0\rangle &
\langle3_1, 3_3, 3_0, 4_0\rangle &
\langle5_4, 4_1, 3_1, 2_4\rangle &
\langle1_0, 5_3, 0_2, 4_0\rangle \\
\langle0_1, 2_0, 3_2, 0_2\rangle &
\langle5_0, 1_3, 2_4, 2_3\rangle &
\langle4_2, 4_0, 4_4, 5_3\rangle &
\langle2_2, 0_1, 4_2, 3_4\rangle &
\langle3_3, 1_4, 0_4, 2_2\rangle &
\langle4_3, 0_3, 5_1, 2_1\rangle \\
\langle0_0, 2_4, 4_4, 5_1\rangle &
\end{array}$$
Then add the following $5$ codewords.
$$\begin{array}{llllll}
\langle0_0, 2_0, 4_0, 0_5\rangle &
\langle1_0, 3_0, 5_0, 1_5\rangle &
\langle0_5, 1_5, 0_6, 1_6\rangle &
\langle0_1, 2_1, 4_1, 0_6\rangle &
\langle1_1, 3_1, 5_1, 1_6\rangle
\end{array}$$

For $t=6$, take a $[3,1]$-GDC$(6)$ of type $15^4$ from
Lemma~\ref{s15^t}. Adjoin one ideal point, and fill in the
groups together with the ideal point with optimal codes of
length $16$ constructed above to obtain the desired code.

Finally, for $t=24$, take a $[3,1]$-GDC$(6)$ of type $18^4$ from
Lemma~\ref{GDC-18^4} and assign each point with weight 3 to obtain a $[3,1]$-GDC$(6)$ of type $54^4$.  Adjoin seven ideal points and fill in the
groups together with these ideal points with $[3,1]$-GDC$(6)$s of type $1^{54}7^1$ from Lemma~\ref{GDC-1^54 7^1} to obtain a $[3,1]$-GDC$(6)$ of type $1^{162} 61^1$. Then fill in the group with an optimal code of length $61$ to obtain the desired code.
\end{proof}

\begin{lemma}
$A_3(9t+7,6,[3,1])=U(9t+7,6,[3,1])$ for each
$t\geq 17$ and $t\not\in \{20,26,28\}$.
\end{lemma}

\begin{proof}
For each $t\equiv 1,2,3\pmod{4}$, $t\geq 17$ and $t\not =26$, write $t=4u+x$ such that $1\leq x\leq 3$. Take a $[3,1]$-GDC$(6)$ of type $36^u(9x)^1$ from Lemmas~\ref{GDCa} and~\ref{GDCb}. Adjoin seven ideal points, and fill in the groups of size $36$ together with these ideal points with $[3,1]$-GDC$(6)$s of type
$1^{36} 7^1$ from Lemma~\ref{GDC-1^36 7^1} to obtain a $[3,1]$-GDC$(6)$ of type $1^{36u}(9x+7)^1$. Then fill in the group with an optimal code of length $9x+7$ from Lemma~\ref{s7(9)} to obtain an optimal
$(9(4u+x)+7,6,[3,1])_3$-code, as desired.

Now we consider the case of $t\equiv 0 \pmod{4}$. For $t=24$, see Lemma~\ref{s7(9)}. For $t\geq 120$, write $t=4u+24$. Take a $[3,1]$-GDC$(6)$ of type $36^u 216^1$ from Lemma~\ref{GDCa}. Adjoin seven ideal points and fill in the groups.
For $32\leq t <120 $, write $t=6u-12+x+y+z$ with  $u\in \{7,9,11,13,16,19\}$ such that $x,y,z \in \{0,2,4,6\}$ and at most one of $x$, $y$, $z$ takes the value 2. Take a TD$(7,u)$ from Theorem~\ref{TD}, and remove one point from one group to obtain a
$\{7,u\}$-GDD of type $6^u(u-1)^1$. Apply the Fundamental
Construction with weight $9$ to the points in $u-2$ groups of
size $6$ and weights $0$ or $9$ to the remaining points. Noting that there exist $[3,1]$-GDC$(6)$s of type $9^s$ for $s\geq 4$ by Theorem~\ref{9^t}, we obtain a $[3,1]$-GDC$(6)$ of type $54^{u-2}(9x)^1 (9y)^1 (9z)^1$. Then adjoin seven ideal points and fill in the groups. Note that there exist $[3,1]$-GDC$(6)$s of types $1^{36} 7^1$ and $1^{54} 7^1$ and optimal codes of lengths $25$ and $61$. We get an optimal
$(9(6u-12+x+y+z)+7,6,[3,1])_3$-code, as desired.
\end{proof}

Summarizing the above results, we have:

\begin{theorem}
\label{7(9)} $A_3(7,6,[3,1])=2$, $A_3(9t+7,6,[3,1])=U(9t+7,6,[3,1])$ for each
$t\geq 1$ and $t\not\in \{4,5\}\cup [7,16] \cup \{20,26,28\}$.
\end{theorem}

\subsection{Case of Length $n\equiv 8\pmod{9}$}

\begin{lemma}
\label{s8(9)a} $A_3(8,6,[3,1])=4$, $A_3(9t+8,6,[3,1])=U(9t+8,6,[3,1])$ for $t\in
\{1,2\}$.
\end{lemma}

\begin{proof} For $t=0$, see Table~\ref{presult}.

For $t=1$, an optimal $(17,6,[3,1])_3$-code is constructed on
$I_{17}$ with $28$ codewords listed below.
$$\begin{array}{llllll}
\langle5,0,4,2\rangle & \langle0,3,16,9\rangle & \langle2,15,10,5\rangle & \langle0,10,6,8\rangle &
\langle12,5,6,15\rangle & \langle7,13,11,5\rangle \\ \langle7,6,9,3\rangle  & \langle1,12,9,0\rangle&
\langle10,7,1,14\rangle &  \langle11,8,4,3\rangle  & \langle16,5,14,7\rangle & \langle8,12,13,14\rangle \\
\langle7,2,8,0\rangle & \langle1,3,5,13\rangle & \langle7,4,15,12\rangle & \langle3,15,14,8\rangle &
\langle9,14,2,13\rangle & \langle15,11,9,16\rangle \\ \langle6,4,14,1\rangle & \langle1,2,11,6\rangle &
\langle13,15,0,1\rangle & \langle4,13,10,9\rangle & \langle1,16,8,15\rangle & \langle14,11,0,12\rangle \\
\langle5,8,9,10\rangle & \langle3,12,2,4\rangle & \langle6,16,13,2\rangle & \langle16,10,12,11\rangle
\end{array}$$

For $t=2$, the desired code is constructed on $X=(\bbZ_{4} \times \{0, 1, 2, 3, 4, 5\})\cup (\bbZ_{2} \times \{6\})$. Define $\alpha: X \rightarrow X$ as $x_y \rightarrow (x+1)_y$ where the addition is modulo $4$ if $y \in \{0, 1, 2, 3, 4, 5\}$, and modulo $2$ if $y \in \{6\}$.
Develop the following $17$ base codewords with $\alpha$:
$$\begin{array}{llllll}
\langle1_3, 2_1, 0_6, 2_5\rangle &
\langle3_5, 3_4, 0_6, 2_3\rangle &
\langle0_4, 1_0, 0_6, 3_1\rangle &
\langle0_0, 2_2, 0_6, 1_2\rangle &
\langle0_3, 3_1, 3_2, 0_6\rangle &
\langle0_3, 3_3, 1_2, 1_0\rangle \\
\langle1_5, 3_2, 3_3, 0_4\rangle &
\langle2_4, 2_3, 1_5, 0_1\rangle &
\langle3_5, 0_5, 0_2, 3_0\rangle &
\langle3_4, 3_1, 0_1, 1_4\rangle &
\langle2_0, 2_1, 2_3, 0_3\rangle &
\langle1_0, 2_3, 1_4, 2_5\rangle \\
\langle0_4, 1_4, 1_2, 3_2\rangle &
\langle3_2, 2_1, 0_5, 0_1\rangle &
\langle2_1, 3_0, 1_5, 3_5\rangle &
\langle0_0, 1_1, 3_2, 0_2\rangle &
\langle0_5, 1_0, 2_4, 0_3\rangle &
\end{array}$$
Further add the codeword $\langle0_0, 1_0, 2_0, 3_0\rangle$.
\end{proof}

\begin{lemma}
\label{GDC-1^9t 8^1} There is a $[3,1]$-GDC$(6)$ of type $1^{9t} 8^1$ with size $9t^2+14t$ for each $t\in \{3,4,5,6\}$.
\end{lemma}

\begin{proof} For each $t$, the desired GDC is constructed on $X=(\bbZ_{3t} \times \{0, 1, 2\})\cup (\bbZ_{3} \times \{3,4\}) \cup \{\infty_0, \infty_1\}$, where the point set $(\bbZ_{3} \times \{3,4\}) \cup \{\infty_0, \infty_1\}$ forms the group of size 8. Define $\alpha: X \rightarrow X$ as $x_y \rightarrow (x+1)_y$ where the addition is modulo $3t$ if $y \in \{0, 1, 2\}$, and modulo 3 if $y \in \{3,4\}$. The points $\infty_0, \infty_1$ are fixed by $\alpha$. First develop the following two codewords with $\alpha$, making two short orbits of length $t$:
$$\begin{array}{ll}
\langle 0_0, t_0, (2t)_0, \infty_0 \rangle &
\langle 0_1, t_1, (2t)_1, \infty_1 \rangle
\end{array}$$
Then develop the following base codewords with $\alpha$:

\noindent $t=3$:
$$\begin{array}{llllll}
\langle0_3, 0_1, 8_1, 7_0\rangle &
\langle0_3, 0_0, 8_0, 1_2\rangle &
\langle0_3, 0_2, 4_1, 5_2\rangle &
\langle8_2, 4_0, 7_2, 0_3\rangle &
\langle0_4, 2_2, 0_1, 1_0\rangle &
\langle0_4, 6_2, 4_2, 5_1\rangle \\
\langle0_4, 0_0, 2_0, 1_1\rangle &
\langle1_0, 7_1, 5_1, 0_4\rangle &
\langle0_0, 0_2, 4_0, 0_1\rangle &
\langle\infty_0, 6_1, 4_2, 6_0\rangle &
\langle6_1, 8_0, 2_1, 6_2\rangle &
\langle3_1, 0_2, 6_2, 7_0\rangle \\
\langle\infty_1, 6_2, 0_0, 2_1\rangle &
\end{array}$$

\noindent $t=4$:
$$\begin{array}{lllll}
\langle\infty_0, 4_1, 11_2, 7_0\rangle &
\langle5_1, 1_2, 6_2, 7_2\rangle &
\langle0_3, 3_1, 10_1, 8_2\rangle &
\langle3_2, 5_2, 8_0, 0_3\rangle &
\langle0_0, 0_1, 1_0, 7_0\rangle \\
\langle2_1, 11_1, 0_1, 11_2\rangle &
\langle0_4, 0_2, 8_2, 6_1\rangle &
\langle5_0, 3_2, 0_1, 10_0\rangle &
\langle0_4, 0_0, 10_1, 4_2\rangle &
\langle1_2, 9_1, 5_0, 0_4\rangle \\
\langle0_4, 10_0, 11_1, 8_0\rangle &
\langle0_3, 2_1, 6_0, 8_0\rangle &
\langle7_2, 5_0, 8_0, 2_1\rangle &
\langle0_3, 4_2, 10_0, 3_2\rangle &
\langle3_0, 3_2, 6_2, 6_1\rangle \\
\langle\infty_1, 9_2, 8_0, 10_1\rangle &
\end{array}$$

\noindent $t=5$:
$$\begin{array}{lllll}
\langle9_1, 5_2, 1_0, 3_2\rangle &
\langle0_1, 3_1, 7_1, 4_2\rangle &
\langle0_4, 2_2, 12_1, 13_1\rangle &
\langle11_0, 6_1, 5_0, 8_0\rangle &
\langle0_3, 11_0, 2_2, 0_0\rangle \\
\langle7_1, 0_0, 9_1, 0_3\rangle &
\langle0_3, 2_1, 9_2, 1_1\rangle &
\langle5_1, 14_0, 11_1, 6_0\rangle &
\langle0_4, 12_2, 5_0, 1_0\rangle &
\langle5_2, 8_2, 10_1, 10_2\rangle \\
\langle0_4, 2_1, 0_0, 1_2\rangle &
\langle4_1, 4_0, 7_2, 0_4\rangle&
\langle2_1, 13_0, 14_0, 10_2\rangle &
\langle11_0, 1_2, 5_2, 4_0\rangle &
\langle\infty_0, 2_1, 4_2, 6_0\rangle \\
\langle0_2, 1_2, 7_2, 1_1\rangle &
\langle0_3, 7_2, 7_0, 6_1\rangle &
\langle\infty_1, 11_0, 10_2, 1_1\rangle &
\langle12_2, 2_0, 4_0, 0_1\rangle  &
\end{array}$$

\noindent $t=6$:
$$\begin{array}{lllll}
\langle17_0, 14_1, 17_1, 13_2\rangle &
\langle0_0, 2_2, 12_1, 9_2\rangle &
\langle14_1, 1_0, 7_1, 12_2\rangle &
\langle3_0, 16_0, 14_0, 0_3\rangle &
\langle0_4, 4_1, 0_1, 13_0\rangle \\
\langle16_2, 10_0, 2_2, 13_2\rangle &
\langle0_4, 0_0, 15_2, 8_1\rangle &
\langle12_1, 4_2, 12_2, 1_0\rangle &
\langle12_1, 13_0, 2_1, 4_0\rangle &
\langle14_1, 2_2, 4_0, 0_4\rangle \\
\langle5_2, 11_2, 10_0, 14_2\rangle &
\langle17_2, 6_1, 0_2, 1_1\rangle &
\langle0_4, 13_2, 14_0, 8_2\rangle &
\langle0_3, 0_1, 17_1, 15_0\rangle &
\langle0_0, 0_2, 10_0, 3_1\rangle \\
\langle15_0, 16_0, 12_0, 17_1\rangle &
\langle0_3, 6_2, 4_2, 10_0\rangle &
\langle\infty_1, 12_2, 5_0, 3_1\rangle &
\langle0_1, 2_2, 16_1, 7_2\rangle &
\langle0_3, 17_2, 4_1, 8_0\rangle \\
\langle\infty_0, 15_1, 0_2, 14_0\rangle &
\langle1_1, 2_2, 15_0, 6_1\rangle &
\end{array}$$
\end{proof}

\begin{lemma}
\label{s8(9)b} $A_3(9t+8,6,[3,1])=U(9t+8,6,[3,1])$ for each $t\geq 13$ and $t\not \in \{15,23,28\}$.
\end{lemma}

\begin{proof}
For each $t\in\{13,16,19\}$, take a $[3,1]$-GDC$(6)$ of type $27^s 9^1$ with $s\in\{4,5,6\}$ from Lemma~\ref{GDC-27^t 9^1} and adjoin eight ideal points. Then fill in the groups together with the ideal points with $[3,1]$-GDC$(6)$s of type $1^{27} 8^1$ and an optimal code of length 17 to obtain the desired code. For each $t\in\{14,20\}$, take a $[3,1]$-GDC$(6)$ of type $27^s 18^1$ with $s\in\{4,6\}$ from Lemma~\ref{GDC-27^t 18^1} and adjoin eight ideal points. Then fill in the groups with $[3,1]$-GDC$(6)$s of type $1^{27} 8^1$ and an optimal code of length 26 to obtain the desired code.
For $t\in \{17,18\}$, take a $[3,1]$-GDC$(6)$ of type $36^4 9^1$ or type $36^4 18^1$ from Lemma~\ref{GDCa}, adjoin eight ideal points and fill in the groups.

For $t\in \{21,22,24,25,26,27\}$, take a TD$(6,5)$ from Theorem~\ref{TD} and apply the Fundamental
Construction to this TD, assigning weight $9$ to the points in the first four groups and weights $0$ or $9$ to the remaining points. Thus we can get a $[3,1]$-GDC$(6)$ of type $45^4 (9x)^1 (9y)^1$ with $x\in\{0,3,4,5\}$ and $y\in\{1,2\}$. Then adjoin eight ideal points and fill in the groups. For $t\in \{29,30\}$, take a $[3,1]$-GDC$(6)$ of type $36^7 9^1$ or type $36^7 18^1$ from Lemma~\ref{GDCa}, adjoin eight ideal points and fill in the groups.

For $31\leq t \leq 80$ and $t\not \in \{33,57\}$, take a TD$(7,u)$ with $u\in \{7,8,9,11,12,13\}$ from Theorem~\ref{TD}, and remove one point from one group to obtain a $\{7,u\}$-GDD of type $6^{u} (u-1)^1$. Apply the Fundamental Construction to this GDD, assigning  weight $9$ to the points in the first $u-2$ groups of size $6$ and weights $0$ or $9$ to the remaining points. We get a $[3,1]$-GDC$(6)$ of type $54^{u-2} (9x)^1 (9y)^1 (9z)^1$ with $x,y\in\{0,3,4,5,6\}$ and $z\in\{1,2\}$. Then adjoin eight ideal points and fill in the groups. For $t\in \{33,57\}$, take a $[3,1]$-GDC$(6)$ of type $36^8 9^1$ or type $36^{14} 9^1$ from Lemma~\ref{GDCa}, adjoin eight ideal points and fill in the groups.

Finally, for $t\geq 80$ and $t\not \in \{83, 87, 91\}$, take a $[3,1]$-GDC$(6)$ of type $36^u (9x)^1$ with $u\geq 16$ and $x\in\{1,2,16\}$ or $u\geq 19$ and $x=19$ from Lemma~\ref{GDCa}. Then adjoin eight ideal points and fill in the groups. For each $t \in \{83, 87, 91\}$, take a TD$(6,u)$ with $u=15$, $17$ or $19$ from Theorem~\ref{TD}, and remove one point from one group to obtain a $\{6,u\}$-GDD of type $5^{u} (u-1)^1$. Apply the Fundamental Construction to this GDD, assigning  weight $9$ to the points in the first $u-1$ groups of size $5$ and weights $0$ or $9$ to the remaining points. Thus we can get a $[3,1]$-GDC$(6)$ of type $45^{14} 117^1$, type $45^{17} 18^1$ or type $45^{18} 9^1$. Then adjoin eight ideal points and fill in the groups to complete the proof.
\end{proof}

Summarizing the above results, we have:

\begin{theorem}
\label{8(9)} $A_3(8,6,[3,1])=4$, $A_3(9t+8,6,[3,1])=U(9t+8,6,[3,1])$ for each
$t\geq 1$ and $t\not\in [3,12] \cup \{15,23,28\}$.
\end{theorem}

\section{Conclusion}

In this paper, we determine almost completely the spectrum of
sizes for optimal ternary constant-composition codes with
weight four and minimum distance six. We summarize our
main results of this paper as follows:

\begin{theorem} For any integer $n \ge 4$, let $Q=\{13, 16, 22, 59, 65, 71, 76, 88, 94, 124\}$,  $Q_1=\{14, 23, 29, 35, 41, 47, 53, 83\} \cup \{n:  95 \leq n \leq 323, n\equiv 11, 17, 23 \pmod{24}\} \cup \{347, 353, 359,$ $371, 377\}$, $Q_2=\{17, 89\}$. Then
\begin{equation*}
A_3(n,6,[2,2])=\begin{cases}
1,&\text{if $n\in \{4,5\}$}; \\
3,&\text{if $n=7$}; \\
5,&\text{if $n=8$}; \\
15, &\text{if $n=11$}; \\
\left\lfloor \frac{n}{2} \left\lfloor \frac{n-1}{3}\right\rfloor \right\rfloor , & \text{if  $n\ge 6$ and $n \not\in \{7,8,11\}\cup Q\cup Q_1 \cup Q_2$.} \\
\end{cases}\\
\end{equation*}
Furthermore, we have
\begin{enumerate}
\item $\left\lfloor \frac{n}{2} \left\lfloor \frac{n-1}{3}\right\rfloor \right\rfloor-1 \le A_3(n,6,[2,2]) \le \left\lfloor \frac{n}{2} \left\lfloor \frac{n-1}{3}\right\rfloor \right\rfloor$ when $n\in Q_1$;
\item $\left\lfloor \frac{n}{2} \left\lfloor \frac{n-1}{3}\right\rfloor \right\rfloor-2 \le A_3(n,6,[2,2]) \le \left\lfloor \frac{n}{2} \left\lfloor \frac{n-1}{3}\right\rfloor \right\rfloor$ when $n\in Q_2$.
\end{enumerate}
\end{theorem}

\begin{theorem} For any integer $n \ge 4$, write $n=9t+i$ with $0\leq i<9$. Then we have:
\begin{enumerate}
\item $i=0$: $A_3(9t,6,[3,1])= 9t^2-3t$ for each $t\geq 1$;
\item $i=1$: $A_3(9t+1,6,[3,1])= 9t^2+t$ for each $t\geq 1,$ except possibly for $t=2$;
\item $i=2$: $A_3(9t+2,6,[3,1])= 9t^2+2t$ for each $t\geq 1$;
\item $i=3$: $A_3(9t+3,6,[3,1])= 9t^2+3t$ for each $t\geq 1$;
\item $i=4$: $A_3(9t+4,6,[3,1])= 9t^2+6t+1+\lfloor\frac{t}{4}\rfloor$ for each $t\geq 0$ and $t\not\in \{1,5,6,7,9,10,11,$ $13,14,15,21,25,26,29,37,41,45,49,53,57,61,65,69,73,77,81\}$;
\item $i=5$: $A_3(9t+5,6,[3,1])= 9t^2+7t+1+\lfloor\frac{t+1}{4}\rfloor$ for each $t\geq 0$ and $t\not\in [5,15]\cup \{26\}$;
\item $i=6$: $A_3(9t+6,6,[3,1])= 9t^2+9t+2$ for each $t\geq 0$;
\item $i=7$: $A_3(9t+7,6,[3,1])= 9t^2+11t+3+\lfloor\frac{t+1}{2}\rfloor$ for each $t\geq 1$ and $t\not\in \{4,5\}\cup [7,16] \cup \{20,26,28\}$;
\item $i=8$: $A_3(9t+8,6,[3,1])= 9t^2+14t+5$ for each $t\geq 1$ and $t\not\in [3,12] \cup \{15,23,28\}$.
\end{enumerate}
\end{theorem}

\subsection*{Acknowledgements}
Hui Zhang is now working as a postdoctoral fellow  with Division~of~Mathematical Sciences, School of Physical and Mathematical Sciences, Nanyang~Technological~University, Singapore~637371, Singapore.

\bibliographystyle{IEEEtran}

\end{document}